\documentclass[reqno]{amsart}

\usepackage[usenames,dvipsnames,svgnames,table]{xcolor}
\usepackage{hyperref}
\usepackage{IEEEtrantools}

\usepackage{amssymb,latexsym,amsfonts,amsmath}
\usepackage{graphicx}
\usepackage{eso-pic}
\usepackage{epsfig}
\usepackage{graphicx}
\usepackage{mathtools}
\usepackage{tikz}
\usetikzlibrary{automata}
\usetikzlibrary{shapes}
\usetikzlibrary{calc,shapes,arrows}
\usepackage{amssymb,latexsym,amsfonts,amsmath}
\usepackage{graphicx}
\usepackage{mdwlist}
\usepackage{refcount}
\usepackage{comment}

\newtheorem{theorem}{Theorem}[section]

\newtheorem{proposition}[theorem]{Proposition}
\newtheorem{corollary}[theorem]{Corollary}

\newtheorem{definition}[theorem]{Definition}
\newtheorem{example}[theorem]{Example}

\newtheorem{remark}[theorem]{Remark}
\numberwithin{equation}{section}

\xdefinecolor{MATLABblue}{rgb}{0.0,0.45,.74}
\xdefinecolor{MATLABred}{rgb}{0.85,0.33,.1}

\newcommand{\R}{{\mathbb{R}}}

\newcommand{\N}{{\mathbb{N}}}

\usepackage[normalem]{ulem}

\usepackage{graphicx,keyval,trig}
\usepackage{amsmath,amssymb,amsfonts,amssymb,dsfont}
\usepackage{mathtools, mathrsfs}

\def\Sig{\Sigma}
\def\e{\epsilon}
\def\a{\alpha}
\DeclareMathOperator{\aug}{aug}
\DeclareMathOperator{\suc}{Succ}
\DeclareMathOperator{\post}{Post}

\begin{document}

\begin{abstract}
	In this paper, we propose several opacity-preserving (bi)simulation relations for general nondeterministic
	transition systems (NTS) in terms of initial-state opacity, current-state opacity,
	$K$-step opacity, and infinite-step opacity. We also show how one can leverage quotient construction to compute such relations.
	In addition, we use a two-way observer method to verify opacity of
	nondeterministic finite transition systems (NFTSs).
	As a result, although the verification of opacity for infinite NTSs is generally \emph{undecidable}, if one can find such an opacity-preserving relation from an infinite NTS to
	an NFTS, the (lack of) opacity of the NTS can be easily verified over the NFTS which is \emph{decidable}.
\end{abstract}

\title[Opacity of nondeterministic transition systems: A (bi)simulation relation approach]{Opacity of nondeterministic transition systems: A (bi)simulation relation approach}

\author[K. Zhang]{Kuize Zhang$^1$}
\author[X. Yin]{Xiang Yin$^2$}
\author[M. Zamani]{Majid Zamani$^3$}
\address{$^1$ACCESS Linnaeus Center, School of Electrical Engineering and Computer Science, KTH Royal Institute of Technology, 10044 Stockholm, Sweden.}
\email{kuzhan@kth.se}
\urladdr{https://www.kth.se/en/eecs}
\address{$^2$Department of Automation, Shanghai Jiao Tong University, Shanghai, China}
\email{xiangyin@umich.edu}
\urladdr{http://automation.sjtu.edu.cn/en}
\address{$^1$Department of Electrical and Computer Engineering, Technical University of Munich, D-80290 Munich, Germany.}
\email{zamani@tum.de}
\urladdr{http://www.hcs.ei.tum.de}

\maketitle

\section{Introduction}

The notion of opacity is introduced in the analysis of cryptographic protocols \cite{Mazare2004Opacity},
and describes the ability that a system forbids leaking secret information.
Given a system, we assume that an intruder (outside the system)
can only observe the external behaviors of the system, i.e., the outputs of the system, but cannot see
the states of the system directly. Then, intuitively the system is called opaque if the intruder cannot determine
whether some states of the system prior to the current time step are secret via observing the outputs
prior to the current time step.

For discrete-event systems (DESs) in the framework of finite automata, the opacity problem has been widely investigated.
In different practical situations, opacity of DESs can be formulated as whether a system can prevent an intruder from
observing whether the initial state (resp., the current state, each state within $K$ steps prior to the current state
for some positive integer $K$,
each state prior to the current state) of the system is secret, i.e., the so-called
initial-state \cite{Saboori2013InitialStateOpacity} (resp. current-state \cite{Saboori2007CurrentStateOpacity},
$K$-step \cite{Saboori2011KStepOpacityJournal}, and infinite-step \cite{Saboori2012InfiniteStepOpacity}) opacity.
It is known that the existing algorithms for verifying these types of opacity have exponential time
complexity (cf. the above references and \cite{Yin2017TWObserverInfiniteStepOpacity}). Unfortunately, it is unlikely that there exist polynomial time algorithms for verifying them
since the problems of determining initial-state opacity, $K$-step opacity, and infinite-step opacity of DESs
are all {PSPACE}-complete \cite{Saboori2013InitialStateOpacity,Saboori2007CurrentStateOpacity,Saboori2011KStepOpacityJournal,Saboori2012InfiniteStepOpacity}.
When the original system is not opaque, several different approaches have also been proposed to enforce opacity;
see, e.g., \cite{takai2008formula,darondeau2014enforcing,wu2014synthesis,zhang2015max,yin2016uniform,tong2017synthesis}.

Nondeterministic transition systems (NTSs), particularly
nondeterministic finite transition systems (NFTSs),
play a fundamental role as a unified modeling framework in the verification and controller synthesis of hybrid systems \cite{tab09,klo08}, and
model checking \cite{bai08}.
Note that for general infinite-state NTSs, the opacity verification problem is undecidable
\cite{Bryans2008OpacityTransitionSystems},
e.g., the initial-state opacity and current-state opacity for labeled Petri nets are undecidable \cite{Tong2017DecidabilityOpacityPetriNets}.
Recently, opacity has also been investigated for other infinite-state systems, e.g.,  pushdown systems \cite{kobayashi2013verification} and
recursive tile systems \cite{chedor2014diagnosis}, where classes of infinite-state systems are identified for which opacity is decidable.
However, for finite-state systems, e.g., finite automata, though {PSPACE}-hard,
the opacity verification problem is always decidable \cite{Saboori2013InitialStateOpacity,Saboori2007CurrentStateOpacity,Saboori2011KStepOpacityJournal,Saboori2012InfiniteStepOpacity}.

Since the opacity verification problem for general NTSs is undecidable and even for NFTSs is {PSPACE}-hard, in this paper we develop a theory based on (bi)simulation relation to verify opacity using (potentially \emph{simpler}) NFTSs. Since the classical notions of (bi)simulation relations \cite{tab09} do not necessarily preserve opacity,
in this framework we first introduce stronger versions of (bi)simulation relations that preserve opacity. As a result, if one can find an NFTS being (bi)simulated by an infinite-state NTS in the sense of the stronger version,
then the opacity of the NTS (undecidable in general) can be verified over the NFTS (decidable). In addition,
if one can find a smaller NFTS being (bi)simulated by a larger NFTS in the sense of the stronger version, then the opacity of the larger NFTS can be efficiently verified over the smaller one.
Particularly, we modify the existing quotient-based construction \cite{tab09} to synthesize quotient systems of NTSs (resp. NFTSs) in terms of proposed opacity-preserving (bi)simulation relations to implement the above idea.

Intuitively, for two NTSs $\Sig_1$ and $\Sig_2$, $\Sig_2$ simulates $\Sig_1$ if each output sequence
generated by $\Sig_1$ can also be generated by $\Sig_2$; $\Sig_2$ bisimulates $\Sig_1$ if
$\Sig_2$ simulates $\Sig_1$ and vice versa (cf. \cite{tab09}).
Usually, (bi)simulation relation can be used to abstract a large-scale system to a smaller one. Then in some sense the smaller system can take place of the larger one in analysis and synthesis
(cf. \cite{girard07,zam12,tab09}).
In this paper, we first define new notions of opacity-preserving (bi)simulation relation,
then we use the proposed notions to give some necessary and sufficient conditions for the opacity of NTSs.
Hence, if one can find an appropriate opacity-preserving (bi)simulation relation
from the original infinite-state NTS $\Sig_1$ to an NFTS $\Sig_2$
(resp. from the original NFTS $\Sig_1$ to an NFTS $\Sig_2$ with remarkably smaller size than that of $\Sig_1$),
then the opacity of $\Sig_1$ can be checked (resp. much faster) by verifying that of
$\Sig_2$. 
In details, we first define a new notion of initial-state opacity-preserving (InitSOP) simulation
relation from one NTS to another NTS, which is actually not the classical simulation relation \cite{tab09}.
Second, because the InitSOP simulation relation does not suffice to preserve the other three types of opacity,
we define also
a notion of infinite-step opacity-preserving (InfSOP) bisimulation relation that preserves
the other three types of opacity and is actually
a stronger version of the classical bisimulation relation \cite{tab09}. 
In addition, we show that under some mild assumptions, the simulation/bisimulation relation from
an NTS to its quotient system becomes InitSOP simulation/InfSOP bisimulation relation, which
provides a constructive scheme for computing opacity-preserving abstractions of NTSs or large NFTSs.
A preliminary investigation of our results on only infinite-step opacity-preserving bisimulation relation appeared in \cite{kuizecdc17}. In this paper we
present a detailed and mature description of the results announced in \cite{kuizecdc17}, including investigating other notions of opacity (initial-state, current-state, and $K$-step opacity).

The remainder of this paper is organized as follows.
In Section \ref{sec:TS}, the basic notions of NTSs/NFTSs and (bi)simulation relation are introduced.
In Section \ref{sec:OpacityPreSimulation}, we show the main results of the paper, i.e, the notions of
opacity-preserving (bi)simulation relations, and their implementation based on quotient systems.
In Section \ref{sec:opacity_relation}, we prove the implication relationship between different notions of
opacity. In Section \ref{sec:checkingNFTS}, we use a two-way observer method to verify different notions of
opacity of NFTSs. In Section \ref{sec:application}, we show how to use the main results of this paper to
verify opacity of an infinite transition system.
Section \ref{sec:conc} concludes the paper.

\section{Preliminaries}\label{sec:TS}
We use the following notations throughout the paper:

\begin{itemize}

\item $\emptyset$: the empty set;

\item $\N$: the set of natural numbers;
\item $\R$: the set of real numbers;
\item $[a,b]:=\{a,a+1,\dots,b\}$, where $a,b\in \N,\;a\le b$;
\item $|X|$: the cardinality of set $X$.
\end{itemize}

NTSs are defined as in \cite{tab09,lin14b} with some modifications to accommodate for secret states.

\begin{definition}\label{d4.1} An NTS $\Sig$ is a septuple $(X,X_0,S,U,\to,Y,h)$ consisting of
\begin{itemize}
\item a (potentially infinite) set $X$ of states,
\item a (potentially infinite) subset  $X_0\subseteq X$ of initial states,
\item a (potentially infinite) subset $S\subseteq X$ of secret states,
\item a (potentially infinite) set $U$ of inputs,
\item a transition relation $\to \subseteq X\times U\times X$,
\item a set $Y$ of outputs, and
\item an output map $h:X\to Y$.
\end{itemize}
\end{definition}

In an NTS, for a state $x\in X$, the output $h(x)$ also means the observation at $x$.
An NTS is called an NFTS if $X$ and $U$ are finite sets.
Elements of $\to$ are called transitions.
Let $X^*$ be the set of strings of finite length over $X$ including the string $\e$ of length $0$
and $X^+$ be $X^*\setminus\{\e\}$.
For each $\xi\in X^*$, $|\xi|$ denotes the length of $\xi$.
For each $\xi\in X^*$, for all integers $0\le i\le j\le|\xi|-1$,
we use $\xi[i,j]$ to denote $\xi(i)\xi(i+1)\dots \xi(j)$ for short. Sets
$U^*,U^+,Y^*$, and $Y^+$ are defined analogously.
Given an input sequence $\a\in U^{*}$, a string $\xi\in X^*$ is called a run over $\a$ if
$|\xi|-1\le|\a|$, $\xi(0)\in X_0$, and
for all $i\in[0,|\xi|-2]$, $(\xi(i),\a(i),\xi(i+1))\in \to$. Particularly, a run $\xi\in X^*$ over input sequence
$\a\in U^*$ is said to be maximal if either $|\xi|-1=|\a|$ or $(\xi(|\xi|-1),\a(|\xi|-1),x')\notin \to$ for any
$x'\in X$. 
For a run $\xi$, $h(\xi(0))\dots h(\xi(|\xi|-1))$ is called an output sequence
generated by the system.
Transitions generated by $\a$ and $\xi$ can be denoted as $\xi(0)\xrightarrow[]{\a(0)}\xi(1)\xrightarrow[]{{\alpha}(1)}
\cdots\xrightarrow[]{\a(|\xi|-2)}\xi(|\xi|-1)$ (or $\xi(0)\xrightarrow[]{\a}\xi(|\xi|-1)$ for short).
A state $x\in X$ is called reachable from an initial state $x_0\in X_0$
if there exists $\a\in U^*$ such that $x_0\xrightarrow[]{\a}x$.
An NTS is called total if for all $x\in X$ and $u\in U$,
there exists $x'\in X$ such that $(x,u,x')\in\to$. Hence, after a total NTS starts running, it never stops.
However, for a non-total NTS, after it starts running, it may stop; and once it stops, it never starts again.
We assume that the termination of running can be observed, and use a new state $\phi$ to denote it.
In order to describe this phenomenon, we extend a non-total NTS $\Sig=(X,X_0,S,U,\to,Y,h)$ to a total NTS
$\Sig_{\aug}:=(X\cup\{\phi\},X_0,S,U,\to_{\aug},Y\cup\{\phi\},h_{\aug})$ as its augmented system,
where $\phi\notin X\cup U\cup Y$, $\to\subseteq\to_{\aug}$, $\to_{\aug}\setminus \to=\{(\phi,u,\phi)|u\in U\}
\cup\{(x,u,\phi)|(x,u,x')\notin\to\text{ for any }x'\in X\}$, $h_{\aug}|_{X}=h$ (i.e., the restriction of
$h_{\aug}$ to $X$ equals $h$), and $h_{\aug}(\phi)=\phi$.
Particularly, for a total NTS, its augmented system, also denoted by $\Sig_{\aug}$, is the NTS itself.

An NTS can be represented by its state transition diagram, i.e., a directed graph whose vertices
correspond to the states and their associated outputs of
the NTS and whose edges correspond to state transitions.
Each edge is labeled with the inputs
associated with the transition, a state directly connected from ``start'' means an initial state, and
a double circle (or rectangle) denotes a secret state.
We give an example to depict these concepts.

\begin{example}\label{exam1_OpacityNFTS}
	Consider NFTS $(X,X_0,S,U,\to,Y,h)$, where $X=\{a,b,c\}$, $X_0=X$, $S=\{b\}$, $U=Y=\{0,1\}$,
	$\to=\{(a,1,a),(a,0,b),(a,0,c),(b,0,b),(b,1,b),(c,0,c),(c,1,b)\}$, $h(a)=0$,
	$h(b)=h(c)=1$ (see Fig. \ref{fig1:detectability_NFTS}).

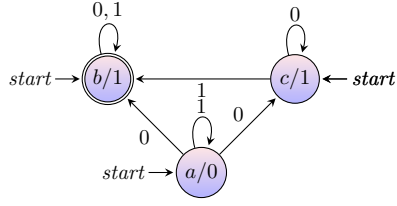
\begin{figure}
        \centering
\begin{tikzpicture}[>=stealth',shorten >=1pt,auto,node distance=2.2 cm, scale = 0.8, transform shape,
	->,>=stealth,inner sep=2pt,state/.style={shape=circle,draw,top color=red!10,bottom color=blue!30},
	point/.style={circle,inner sep=0pt,minimum size=2pt,fill=}, 
	skip loop/.style={to path={-- ++(0,#1) -| (\tikztotarget)}}]
	\node[initial,state] (a)                                 {$a/0$};
	\node[initial,state,accepting] (b) [above left of = a]                    {$b/1$};
	\node[initial,state,initial by arrow,initial where=right] (c) [above right of = a]                 {$c/1$};
	\node[] (virtual) [right of = c] {};

	\path [->] (a) edge [loop above] node {$1$} (a)
	(a) edge node {$0$} (c)
	(a) edge node {$0$} (b)
	(b) edge [loop above] node {$0,1$} (b)
	(c) edge [loop above] node {$0$} (c)
	(c) edge node {$1$} (b)
	;

        \end{tikzpicture}
		\caption{State transition diagram of the NFTS in Example \ref{exam1_OpacityNFTS}.}
	\label{fig1:detectability_NFTS}
\end{figure}
\end{example}

Here, we recall the classical notions of (bi)simulation relations (see for example, \cite{tab09}).
\begin{definition}[simulation]\label{def_simulation}
	Consider two NTSs $\Sig_i=(X_i,X_{i,0},S_i,U_i,\to_i,Y,h_i)$, $i=1,2$.
	A relation $\sim\subseteq X_1\times X_2$ is called a simulation relation from $\Sig_1$ to $\Sig_2$ if
	\begin{enumerate}
		\item\label{item1_simulation}
			for every $x_{1,0}\in X_{1,0}$, there exists $x_{2,0}\in X_{2,0}$ such that $(x_{1,0},x_{2,0})\in\sim$;
		\item\label{item2_simulation}
			for every $(x_1,x_2)\in\sim$, $h_1(x_1)=h_2(x_2)$;
		\item\label{item3_simulation}
			for every $(x_1,x_2)\in\sim$, if there is a transition $x_1\xrightarrow[]{u_1}_1 x_1'$ in $\Sig_1$
			then there exists a transition $x_2\xrightarrow[]{u_2}_2 x_2'$ in $\Sig_2$ satisfying $(x_1',x_2')\in\sim$.
	\end{enumerate}
	Under a simulation relation $\sim\subseteq X_1\times X_2$ from $\Sig_1$ to $\Sig_2$, we say $\Sig_2$ simulates $\Sig_1$,
	and denote it by $\Sig_1\preceq_\mathsf{S}\Sig_2$.
\end{definition}

\begin{definition}[bisimulation]\label{def_bisimulation}
	Consider two NTSs $\Sig_i=(X_i,X_{i,0},S_i,U_i,\to_i,Y,h_i)$, $i=1,2$.
	A relation $\sim\subseteq X_1\times X_2$ is called a bisimulation relation between $\Sig_1$ and $\Sig_2$ if
	\begin{enumerate}
		\item\label{item1_bisimulation}
			\begin{enumerate}
				\item\label{item1a_bisimulation}
					for every $x_{1,0}\in X_{1,0}$, there exists $x_{2,0}\in X_{2,0}$ such that $(x_{1,0},x_{2,0})\in\sim$;
				\item\label{item1b_bisimulation}
					for every $x_{2,0}\in X_{2,0}$, there exists $x_{1,0}\in X_{1,0}$ such that $(x_{1,0},x_{2,0})\in\sim$;
			\end{enumerate}
			
		\item\label{item2_bisimulation}
			for every $(x_1,x_2)\in\sim$, $h_1(x_1)=h_2(x_2)$;
		\item\label{item3_bisimulation}
			for every $(x_1,x_2)\in\sim$,
			\begin{enumerate}
				\item\label{item3a_bisimulation}
					if there exists a transition $x_1\xrightarrow[]{u_1}_1 x_1'$ in $\Sig_1$
					then there exists a transition $x_2\xrightarrow[]{u_2}_2 x_2'$ in $\Sig_2$ satisfying $(x_1',x_2')\in\sim$;
				\item\label{item3b_bisimulation}
					if there exists a transition $x_2\xrightarrow[]{u_2}_2 x_2'$ in $\Sig_2$
					then there exists a transition $x_1\xrightarrow[]{u_1}_1 x_1'$ in $\Sig_1$ satisfying $(x_1',x_2')\in\sim$.
			\end{enumerate}
	\end{enumerate}
	Under a bisimulation relation $\sim\subseteq X_1\times X_2$ between $\Sig_1$ and $\Sig_2$, we say $\Sig_2$ bisimulates $\Sig_1$ and vice versa, and denote it by
	$\Sig_1\cong_\mathsf{S}\Sig_2$.
\end{definition}

From Definitions \ref{def_simulation} and \ref{def_bisimulation}, one can readily see that if
$\Sig_2$ simulates  $\Sig_1$ then each output sequence generated by $\Sig_1$ can be generated by $\Sig_2$ as well;
and if $\Sig_2$ bisimulates $\Sig_1$ then the set of output sequences generated by $\Sig_1$ coincides with
that generated by $\Sig_2$.

Here, we recall notions of quotient relation and quotient systems \cite{tab09}
with some modifications which will be used later to show one of the main results of the paper.

\begin{definition}(Quotient system)\label{def_quotientsystem}
	Let $\Sig=(X,X_0,S,U,\to,Y,h)$ be an NTS and $\sim\subseteq X\times X$ an equivalence relation on $X$
	satisfying $h(x)=h(x')$ for all $(x,x')\in\sim$.
	The quotient system of $\Sig$ by $\sim$, denoted by $\Sig_{\sim}$, is defined as the system
 	$\Sig_{\sim}=(X_{\sim},X_{\sim,0},S_{\sim},U,\to_{\sim},Y,h_{\sim})$ satisfying
	\begin{enumerate}
		\item $X_{\sim}=X/\sim=\{[x]|x\in X\}$;
		\item $X_{\sim,0}=\{[x]|x\in X,[x]\cap X_0\ne\emptyset\}$;
		\item $S_{\sim}=\{[x]|x\in X,[x]\cap S\ne\emptyset\}$;
		\item for all $[x],[x']\in X_{\sim}$ and $u\in U$, there exists transition
			$[x]\xrightarrow[]{u}_{\sim}[x']$ in $\Sig_{\sim}$ if and only if there exists transition
			$\bar x\xrightarrow[]{u}\bar x'$ in $\Sig$ for some $\bar x\in [x]$ and $\bar x'\in [x']$;
		\item $h_{\sim}([x])=h(\bar x)$ for every $\bar x\in [x]$;
	\end{enumerate}
	where for every $x\in X$, $[x]$ denotes the equivalence class generated by $x$, i.e.,
	$[x]:=\{x'\in X|(x',x)\in\sim\}$.
\end{definition}

It can be seen that for all $x,x'\in X$, 1) either $[x]=[x']$ or $[x]\cap[x']=\emptyset$;
2) $x\in[x']$ if and only if $[x]=[x']$.
Then the set of all distinct equivalence classes corresponding to $\sim$ partitions $X$.
Note that in \cite{tab09}, there is no item for $S_{\sim}$, since the system $\Sig$
considered in \cite{tab09} does not have secret states.
From Definition \ref{def_quotientsystem}, one can easily verify that the number of states in the quotient system $\Sig_{\sim}$ is less than or equal to that in $\Sig$.

Consider an NTS $\Sig=(X,X_0,S,U,\to,Y,h)$ and its quotient system
$\Sig_{\sim}=(X_{\sim},X_{\sim,0},S_{\sim},U,\to_{\sim},Y,h_{\sim})$ defined
by an equivalence relation $\sim\subseteq X\times X$ satisfying $h(x)=h(x')$
for all $(x,x')\in\sim$. By defining a quotient relation
\begin{equation}
	\sim_\mathsf{Q}:=\{(x,[x])|x\in X\}\subseteq X\times X_{\sim},
	\label{eqn:quotientrelation}
\end{equation}
the following result, borrowed from \cite{tab09}, holds.

\begin{proposition}\label{prop_quotientsimulation}
	Consider an NTS $\Sig=(X,X_0,S,U,\to,Y,h)$ and its quotient system
	$\Sig_{\sim}=(X_{\sim},X_{\sim,0},S_{\sim},U,\to_{\sim},Y,h_{\sim})$ defined
	by an equivalence relation $\sim\subseteq X\times X$ satisfying $h(x)=h(x')$
	for all $(x,x')\in\sim$. Under quotient relation $\sim_{\mathsf{Q}}$ defined in \eqref{eqn:quotientrelation},
	$\Sig_{\sim}$ simulates $\Sig$. Moreover, $\Sig_{\sim}$ bisimulates $\Sig$ under $\sim_{\mathsf{Q}}$
	if and only if $\Sig$ bisimulates $\Sig$ under $\sim$.
\end{proposition}


In the sequel, with these preliminaries, we present our main results.

\section{Opacity-Preserving (Bi)simulation Relations}\label{sec:OpacityPreSimulation}

\subsection{Concepts of opacity}

In this subsection, we formulate the notions of opacity of NTSs.

\begin{definition}[InitSO]\label{def_initial-state_opacity}
	Let $\Sig=(X,X_0,S,U,\to,Y,h)$ be an NTS. System $\Sig$ is said to be initial-state opaque if
	for every $x_0\in X_0\cap S$, every $\a\in U^*$, and every maximal run $x_0x_1\dots x_{k}\in X^*$ over $\a$
	 with $k\le|\a|$,
	there exists a maximal run $x_0'\dots x_{k}'\in X^*$ also over $\a$ such that
	$x_0'\notin S$, and $h(x_j)=h(x_j')$ for every $j\in[0,k]$.
\end{definition}

Intuitively, if a system $\Sig$ is initial-state opaque, then the intruder cannot make sure whether the initial
state is secret or not. 

\begin{definition}[CSO]\label{def_current-state_opacity}
	Let $\Sig=(X,X_0,S,U,\to,Y,h)$ be an NTS. System $\Sig$ is said to be current-state opaque if
	for every $x_0\in X_0$, every $\a\in U^*$, and every run $x_0x_1\dots x_{|\a|}\in X^*$ over $\a$,
	if $x_{|\a|}\in S$ then   there exists a run $x_0'\dots x_{|\a|}'\in X^*$ also over $\a$ such that
	$x_{|\a|}'\notin S$, and $h(x_j)=h(x_j')$ for every $j\in[0,|\a|]$.
\end{definition}

Intuitively, if a system $\Sig$ is current-state opaque, then the intruder cannot make sure whether
the current state is secret. 

\begin{definition}[KSO]\label{def_K-step_opacity}
	Let $\Sig=(X,X_0,S,U,\to,Y,h)$ be an NTS.
	System $\Sig$ is said to be $K$-step opaque for a given positive integer $K$ if
	for every $x_0\in X_0$, every $\a\in U^*$, every run $x_0x_1\dots x_{|\a|}\in X^*$ over $\a$,
	and every $i\in[K',|\a|]$,
	if $x_i\in S$ then there exists a run $x_0'\dots x_{|\a|}'\in X^*$ also
	over $\a$ such that $x_i'\notin S$,
	and $h(x_j)=h(x_j')$ for every $j\in[0,|\a|]$, where $K'=\max\{0,|\a|-K\}$.
\end{definition}

\begin{definition}[InfSO]\label{def_infinite-step_opacity}
	Let $\Sig=(X,X_0,S,U,\to,Y,h)$ be an NTS. System $\Sig$ is said to be infinite-step opaque if
	for every $x_0\in X_0$, every $\a\in U^*$, every maximal run $x_0x_1\dots x_{k}\in X^*$ over $\a$
	with $k\le|\a|$,
	and every $i\in[0,k]$,
	if $x_i\in S$ then there exists a maximal run $x_0'\dots x_{k}'\in X^*$ also over $\a$ such that
	$x_i'\notin S$, and $h(x_j)=h(x_j')$ for every $j\in[0,k]$.
\end{definition}

Intuitively, if a system $\Sig$ is infinite (resp. $K$)-step opaque, then the intruder cannot make sure whether any
state (within $K$ steps) prior to the current state is secret.

It is readily seen that an NTS $\Sig$ is initial-state (resp. current-state, $K$-step, infinite-step) opaque
if and only if its augmented system $\Sig_{\aug}$ is initial-state (resp. current-state, $K$-step, infinite-step)
opaque.  Hence, without loss of generality, we can consider only total NTSs in what follows.

\subsection{Initial-state opacity-preserving (bi)simulation relation}

In this subsection, we characterize the initial-state opacity-preserving (InitSOP) simulation relation.

One of the main goals of this subsection is to provide a simulation-based method for verifying
the initial-state opacity of NTSs.
Particularly, for two NTSs $\Sig_1$ and $\Sig_2$, we are interested in providing a new notion of simulation
relation such that $\Sig_2$ simulating $\Sig_1$ implies that if $\Sig_1$ is initial-state opaque then
$\Sig_2$ is also initial-state opaque. In other words, lack of opacity in $\Sigma_2$ implies lack of opacity in $\Sigma_1$.
Hence, the central problem is whether the classical
simulation relation preserves initial-state opacity.
We next show that generally the classical simulation relation does not preserve initial-state opacity.

\begin{proposition}\label{prop:SimulationOpacity_initial}
	Simulation relation (cf. Definition \ref{def_simulation}) does not preserve initial-state opacity.
\end{proposition}

\begin{proof}
	We provide a counterexample to prove the statement.
	Consider two NFTSs $\Sig_i=(X_i,X_{i,0},S_i,U,\to_i,Y,h_i)$, $i=1,2$,
	shown in Fig. \ref{fig1:initial_opaicty_NFTS}, where $X_1=\{1',2',3',4'\}=X_{1,0}$, $S_1=\{1'\}$,
	$X_2=\{1,2\}=X_{2,0}$, $S_2=\{1\}=U$, $Y=\{1,2\}$.

	By Definition \ref{def_initial-state_opacity},
	system $\Sig_1$ is initial-state opaque,
	because for input sequence $\a:=1\dots1\in U^*$, for run $x_1:=1'2'3'\dots$ over $\a$,
	there is a unique run $x_2:=3'4'1'\dots$ over $\a$ such that they produce the same output sequence
	$121\dots$,
	where $1'\in X_{1,0}\cap S_1$ and $3'\in X_{1,0}\setminus S_1$ are initial states.
	Again by Definition \ref{def_initial-state_opacity}, system $\Sig_2$ is not initial-state opaque,
	because for secret state $1$, there exists no other state producing the same output as $1$.
	On the other hand, it can be readily verified that under relation $\sim=\{(1',1),(2',2),(3',1),(4',2)\}$,
	$\Sig_2$ simulates $\Sig_1$. Hence, simulation relation does not preserve initial-state opacity. Similarly, one can readily show that relation $\sim^{-1}=\{(1,1'),(2,2'),(1,3'),(2,4')\}$ is a simulation relation from $\Sigma_2$ to $\Sigma_1$. Hence, the simulation relation does not preserve the lack of initial-state opacity either.
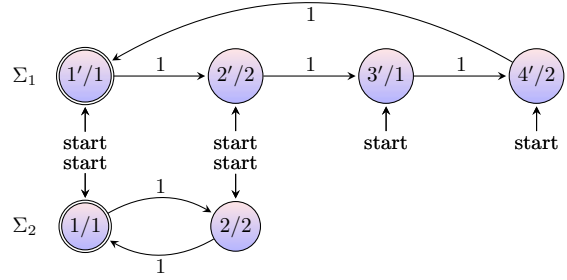
\begin{figure}
        \centering
\begin{tikzpicture}[>=stealth',shorten >=1pt,auto,node distance=2.5 cm, scale = 0.8, transform shape,
	->,>=stealth,inner sep=2pt,state/.style={shape=circle,draw,top color=red!10,bottom color=blue!30},
	point/.style={circle,inner sep=0pt,minimum size=2pt,fill=}, 
	skip loop/.style={to path={-- ++(0,#1) -| (\tikztotarget)}}]
	\node[initial,state,accepting,initial by arrow,initial where=below] (1') {$1'/1$};
	\node[initial,state,initial by arrow,initial where=below] (2') [right of = 1'] {$2'/2$};
	\node[initial,state,initial by arrow,initial where=below] (3') [right of = 2'] {$3'/1$};
	\node[initial,state,initial by arrow,initial where=below] (4') [right of = 3'] {$4'/2$};
	\node[initial,state,accepting,initial by arrow,initial where=above] (1) [below of = 1']  {$1/1$};
	\node[initial,state,initial by arrow,initial where=above] (2) [right of = 1] {$2/2$};

	\path [->]
	(1) edge [bend left] node {$1$} (2)
	(2) edge [bend left] node {$1$} (1)
	(1') edge node {$1$} (2')
	(2') edge node {$1$} (3')
	(3') edge node {$1$} (4')
	(4') edge [bend right] node {$1$} (1')
	;

	\node at (-1,0) {$\Sig_1$};
	\node at (-1,-2.5) {$\Sig_2$};
        \end{tikzpicture}
		\caption{State transition diagrams of two NFTSs in the proof of Prop.~\ref{prop:SimulationOpacity_initial}.}
	\label{fig1:initial_opaicty_NFTS}
\end{figure}
\end{proof}


Since generally simulation relation does not preserve (lack of) initial-state opacity,
we propose a variant of this notion to make it initial-state opacity-preserving.

\begin{definition}[InitSOP simulation relation]\label{def_OpacityPreservingsimulation_initial}
	Consider two NTSs $\Sig_i=(X_i,X_{i,0},S_i,U,\to_i,Y,h_i)$, $i=1,2$.
	A relation $\sim\subseteq X_1\times X_2$ is called an IntiSOP simulation relation from $\Sig_1$ to $\Sig_2$  if
	\begin{enumerate}
		\item\label{item1:opacity_initial}
			\begin{enumerate}
				\item\label{item1a:opacity_initial}
					for all $x_{1,0}\in X_{1,0}\setminus S_1$, there exists $x_{2,0}\in X_{2,0}\setminus S_2$
					such that $(x_{1,0},x_{2,0})\in \sim$;
				\item\label{item1c:opacity_initial}
					for all $x_{2,0}\in X_{2,0} \cap S_2$, there exists $x_{1,0}\in  X_{1,0}\cap S_1$
					such that $(x_{1,0},x_{2,0})\in \sim$;
			\end{enumerate}
		\item\label{item3:opacity_initial} for every $(x_1,x_2)\in\sim$, $h_1(x_1)=h_2(x_2)$;	
		\item\label{item2:opacity_initial} for every $(x_1,x_2)\in\sim$,
			\begin{enumerate}
				\item\label{item2a:opacity_initial}
					for every transition $x_1\xrightarrow[]{u}_1 x_1'$, there exists transition
					$x_2\xrightarrow[]{u}_2 x_2'$ such that $(x_1',x_2')\in\sim$;
				\item\label{item2c:opacity_initial}
					for every transition $x_2\xrightarrow[]{u}_2 x_2'$, there exists transition
					$x_1\xrightarrow[]{u}_1 x_1'$ such that $(x_1',x_2')\in\sim$.
			\end{enumerate}
	\end{enumerate}
\end{definition}

Note that \ref{item1:opacity_initial}) of Definition \ref{def_OpacityPreservingsimulation_initial}
and \ref{item1_simulation}) of Definition \ref{def_simulation} are not comparable. Hence, Definition \ref{def_OpacityPreservingsimulation_initial} is not the classical simulation relation.
Note also that \ref{item2:opacity_initial}) of Definition \ref{def_OpacityPreservingsimulation_initial}
is stronger than \ref{item3_simulation}) of Definition \ref{def_simulation}. Though stronger, \ref{item2:opacity_initial})
of Definition \ref{def_OpacityPreservingsimulation_initial} is somehow necessary for
preserving initial-state opacity.

\begin{theorem}\label{thm1:opacityNFTS_initial}
	Consider two NTSs $\Sig_i=(X_i,X_{i,0},S_i,U,\to_i,Y,h_i)$, $i=1,2$.
	Assume that there exists an InitSOP simulation relation $\sim\subseteq X_1\times X_2$ from $\Sig_1$ to $\Sig_2$.
	If $\Sig_1$ is initial-state opaque
	then $\Sig_2$ is also initial-state opaque.
\end{theorem}

\begin{proof}
	Assume there exists an InitSOP simulation relation $\sim\subseteq X_1\times X_2$ from $\Sig_1$ to $\Sig_2$
	and system $\Sig_1$ is initial-state opaque. Next we prove that $\Sig_2$ is also initial-state opaque.

	For system $\Sig_2$, we arbitrarily choose input sequence $\a\in U^*$, states $x_{2,0},
	x_{2,1},\dots,x_{2,|\a|}\in X_2$ such that
	\begin{equation}\label{eqn3:opacity}
		x_{2,0}\xrightarrow[]{\a(0)}_2 x_{2,1}\xrightarrow[]{\a(1)}_2\cdots\xrightarrow[]{\a(|\a|-1)}_2
		x_{2,|\a|},
	\end{equation}
	and $x_{2,0}\in X_{2,0}\cap S_2$.

	By Definition \ref{def_OpacityPreservingsimulation_initial}
	(specially by conditions \ref{item2a:opacity_initial}) and \ref{item2c:opacity_initial})),
	there exist $x_{1,0}\in X_{1,0}\cap S_1$, $x_{1,j}\in X_1$, $j\in[1,|\a|]$
	such that $h_1(x_{1,k})=h_2(x_{2,k})$ for all $k$ in $[0,|\a|]$, and
	\begin{equation}\label{eqn4:opacity}
		x_{1,0}\xrightarrow[]{\a(0)}_1 x_{1,1}\xrightarrow[]{\a(1)}_1\cdots\xrightarrow[]{\a(|\a|-1)}_1
	x_{1,|\a|}.
	\end{equation}

	Since $\Sig_1$ is initial-state opaque, there exist $x_{1,0}'\in X_{1,0}\setminus S_1$, $x_{1,j}'\in X_1$,
	$j\in[1,|\a|]$ such that $h_1(x_{1,k})=h_1(x_{1,k}')$ for all $k$ in $[0,|\a|]$, and
	$x_{1,0}'\xrightarrow[]{\a(0)}_1 x_{1,1}'\xrightarrow[]{\a(1)}_1\cdots\xrightarrow[]{\a(|\a|-1)}_1
	x_{1,|\a|}'$.

	Also by Definition \ref{def_OpacityPreservingsimulation_initial},
	there exist $x_{2,0}'\in X_{2,0}\setminus S_2$
	and $x_{2,1}',\dots,x_{2,|\a|}'\in X_2$ such that
	$h_1(x_{1,k}')=h_2(x_{2,k}')$, for all $k\in[0,|\a|]$, and
	$x_{2,0}'\xrightarrow[]{\a(0)}_2 x_{2,1}'\xrightarrow[]{\a(1)}_2\cdots\xrightarrow[]{\a(|\a|-1)}_2
	x_{2,|\a|}'$.  Hence, $\forall j\in[0,|\a|]:h_2(x_{2,j})=h_2(x_{2,j}')$, and $\Sig_2$ is initial-state opaque.
\end{proof}

In Definition~\ref{def_OpacityPreservingsimulation_initial}, in addition to requiring equivalent observation at two related states, i.e., condition~2),
we also have four conditions \ref{item1a:opacity_initial}), \ref{item1c:opacity_initial}), \ref{item2a:opacity_initial}), and \ref{item2c:opacity_initial}).
In particular,  conditions \ref{item2a:opacity_initial})  and \ref{item2c:opacity_initial}) are similar to  those in the standard bisimulation relation.
The question then arises as why we need such  strong conditions for  InitSOP \emph{simulation} relation.
In the next four examples, we show that these conditions are all necessary to make it initial-state opacity-preserving even for one direction.

\begin{example}\label{rem1:initial_opaicty_NFTS}
	Recall the NFTSs shown in Fig. \ref{fig1:initial_opaicty_NFTS}. We showed that
	$\Sig_2$ simulates $\Sig_1$, $\Sig_1$ is initial-state opaque, but $\Sig_2$ not.
	We directly see that the simulation relation $\sim=\{(1',1),(2',2),(3',1),(4',2)\}$ in the proof of
	Proposition \ref{prop:SimulationOpacity_initial} from $\Sig_1$ to $\Sig_2$ does not satisfy
	\ref{item1a:opacity_initial}) of Definition \ref{def_OpacityPreservingsimulation_initial}, since
	for state $3'\in X_{1,0}\setminus S_1$, the unique state $1$ satisfying $(3',1)\in\sim$ does not belong to
	$X_{2,0}\setminus S_2$. We also see that relation $\sim$ satisfies all other items of Definition
	\ref{def_OpacityPreservingsimulation_initial}. Hence, \ref{item1a:opacity_initial}) in
	Definition \ref{def_OpacityPreservingsimulation_initial} is necessary to make it initial-state opacity-preserving.
\end{example}

\begin{example}\label{rem2:initial_opaicty_NFTS}
	Consider two NFTSs
	$\Sig_i=(X_i,X_{i,0},S_i,U,\to_i,Y,h_i)$, $i=1,2$,
	shown in Fig. \ref{fig2:initial_opaicty_NFTS}, where $X_1=\{1,2,3,4,5,6\}$, $X_{1,0}=
	\{1,2,3,4\}$, $S_1=\{1\}$;
	$X_2=\{1',2',3',4',5',6'\}=X_{2,0}$, $S_2=\{5',6'\}$, $U=\{1\}$, $Y=\{1,2,3\}$.
	For system $\Sig_1$, it can be verified that relation $\{(1,3),(3,1),(2,4),(4,2),(5,6),(6,5)\}$
	is a bisimulation relation between $\Sig_1$ and itself, then for each run starting from state $1$,
	there is a run starting from state $3$ such that these two runs produce the same output sequence,
	i.e., $\Sig_1$ is initial-state opaque. It is evident that system $\Sig_2$ is not initial-state opaque,
	since if the initial output is $3$ then one knows that the initial states of $\Sig_2$ are secret.
	Now consider relation $\sim=\{(1,1'),(2,2'),(3,3'),(4,4'),(5,5'),(6,6')\}$. One can verify that $\sim$
	satisfies all items of Definition \ref{def_OpacityPreservingsimulation_initial} other than
	\ref{item1c:opacity_initial}). Hence, \ref{item1c:opacity_initial}) in Definition \ref{def_OpacityPreservingsimulation_initial}
	is also necessary to make it initial-state opacity-preserving.

	\begin{figure}
        \centering
\begin{tikzpicture}[>=stealth',shorten >=1pt,auto,node distance=2.5 cm, scale = 0.8, transform shape,
	->,>=stealth,inner sep=2pt,state/.style={shape=circle,draw,top color=red!10,bottom color=blue!30},
	point/.style={circle,inner sep=0pt,minimum size=2pt,fill=}, 
	skip loop/.style={to path={-- ++(0,#1) -| (\tikztotarget)}}]
	\node[initial,state,accepting,initial by arrow,initial where=below] (1) {$1/1$};
	\node[initial,state,initial by arrow,initial where=below] (2) [right of = 1] {$2/2$};
	\node[initial,state,initial by arrow,initial where=below] (3) [right of = 2] {$3/1$};
	\node[initial,state,initial by arrow,initial where=below] (4) [right of = 3] {$4/2$};
	\node[state] (5) [above of = 1] {$5/3$};
	\node[state] (6) [right of = 5] {$6/3$};
	\node[initial,state,initial by arrow,initial where=above] (1') [below of = 1]  {$1'/1$};
	\node[initial,state,initial by arrow,initial where=above] (2') [right of = 1'] {$2'/2$};
	\node[initial,state,initial by arrow,initial where=above] (3') [right of = 2'] {$3'/1$};
	\node[initial,state,initial by arrow,initial where=above] (4') [right of = 3'] {$4'/2$};
	\node[initial,state,accepting,initial by arrow,initial where=below] (5') [below of = 1'] {$5'/3$};
	\node[initial,state,accepting,initial by arrow,initial where=below] (6') [right of = 5'] {$6'/3$};

	\path [->]
	(1') edge node {$1$} (2')
	(2') edge node {$1$} (3')
	(3') edge node {$1$} (4')
	(4') edge [bend left] node {$1$} (1')
	(1') edge node {$1$} (5')
	(5') edge (1')
	(5') edge node {$1$} (6')
	(6') edge (5')
	(6') edge node {$1$} (3')
	(3') edge (6')
	;

	\path [->]
	(1) edge node {$1$} (2)
	(2) edge node {$1$} (3)
	(3) edge node {$1$} (4)
	(4) edge [bend right] node {$1$} (1)
	(1) edge node {$1$} (5)
	(5) edge (1)
	(5) edge node {$1$} (6)
	(6) edge (5)
	(6) edge node {$1$} (3)
	(3) edge (6)
	;

	\node at (-1,0) {$\Sig_1$};
	\node at (-1,-2.5) {$\Sig_2$};
        \end{tikzpicture}
		\caption{State transition diagrams of two NFTSs in Example \ref{rem2:initial_opaicty_NFTS}.}
	\label{fig2:initial_opaicty_NFTS}
\end{figure}
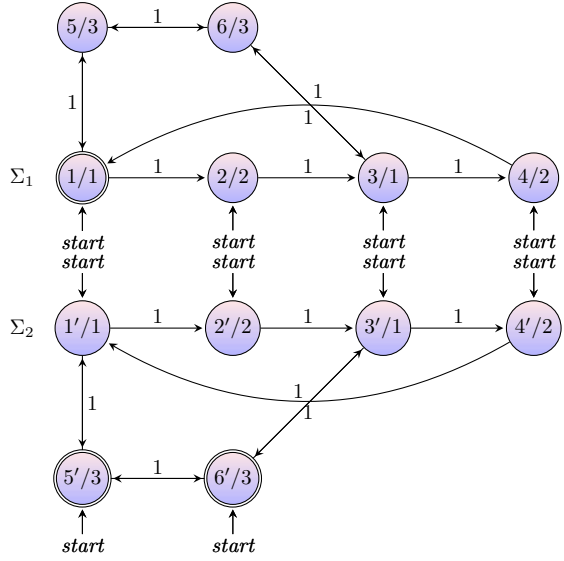

\end{example}

\begin{example}\label{rem4:initial_opaicty_NFTS}
	Consider two NFTSs
	$\Sig_i=(X_i,X_{i,0},S_i,U,\to_i,Y,h_i)$, $i=1,2$,
	shown in Fig. \ref{fig4:initial_opaicty_NFTS}, where $X_1=\{1,2,3,4\}=X_{1,0}$,
	$S_1=\{1\}$;
	$X_2=\{1',2',3',4'\}=X_{2,0}$, $S_2=\{1'\}$, $U=\{1,2\}$, $Y=\{1,2\}$.
	For system $\Sig_1$, it is directly obtained that for each input sequence $\a\in U^*$,
	and each run starting from state $1$ over $\a$, there is a run starting from state $3$ also over $\a$,
	i.e., $\Sig_1$ is initial-state opaque. For system $\Sig_2$, consider input sequence $2$
	and run $1'2'$ over input sequence $2$. However, there is no run starting from $3'$ over input sequence $2$,
	impling that $\Sig_2$ is not initial-state opaque.
	Now consider relation $\sim=\{(1,1'),(2,2'),(3,3'),(4,4')\}$. One can show that $\sim$
	satisfies all items of Definition \ref{def_OpacityPreservingsimulation_initial} other than
	\ref{item2a:opacity_initial}). Therefore, \ref{item2a:opacity_initial}) in Definition \ref{def_OpacityPreservingsimulation_initial}
	is also necessary to make it initial-state opacity-preserving.

	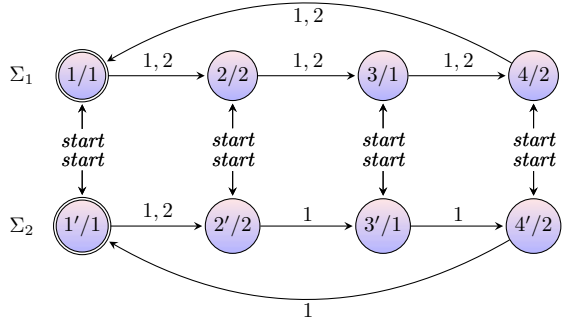
\begin{figure}
        \centering
\begin{tikzpicture}[>=stealth',shorten >=1pt,auto,node distance=2.5 cm, scale = 0.8, transform shape,
	->,>=stealth,inner sep=2pt,state/.style={shape=circle,draw,top color=red!10,bottom color=blue!30},
	point/.style={circle,inner sep=0pt,minimum size=2pt,fill=}, 
	skip loop/.style={to path={-- ++(0,#1) -| (\tikztotarget)}}]
	\node[initial,state,accepting,initial by arrow,initial where=below] (1) {$1/1$};
	\node[initial,state,initial by arrow,initial where=below] (2) [right of = 1] {$2/2$};
	\node[initial,state,initial by arrow,initial where=below] (3) [right of = 2] {$3/1$};
	\node[initial,state,initial by arrow,initial where=below] (4) [right of = 3] {$4/2$};
	\node[initial,state,accepting,initial by arrow,initial where=above] (1') [below of = 1]  {$1'/1$};
	\node[initial,state,initial by arrow,initial where=above] (2') [right of = 1'] {$2'/2$};
	\node[initial,state,initial by arrow,initial where=above] (3') [right of = 2'] {$3'/1$};
	\node[initial,state,initial by arrow,initial where=above] (4') [right of = 3'] {$4'/2$};

	\path [->]
	(1') edge node {$1,2$} (2')
	(2') edge node {$1$} (3')
	(3') edge node {$1$} (4')
	(4') edge [bend left] node {$1$} (1')
	;

	\path [->]
	(1) edge node {$1,2$} (2)
	(2) edge node {$1,2$} (3)
	(3) edge node {$1,2$} (4)
	(4) edge [bend right] node {$1,2$} (1)
	;

	\node at (-1,0) {$\Sig_1$};
	\node at (-1,-2.5) {$\Sig_2$};
        \end{tikzpicture}
		\caption{State transition diagrams of two NFTSs in Example \ref{rem4:initial_opaicty_NFTS}.}
	\label{fig4:initial_opaicty_NFTS}
\end{figure}

\end{example}

\begin{example}\label{rem3:initial_opaicty_NFTS}
	Consider two NFTSs
	$\Sig_i=(X_i,X_{i,0},S_i,U,\to_i,Y,h_i)$, $i=1,2$,
	shown in Fig. \ref{fig3:initial_opaicty_NFTS}, where $X_1=\{1,2,3,4\}=X_{1,0}$,
	$S_1=\{1\}$;
	$X_2=\{1',2',3',4'\}=X_{2,0}$, $S_2=\{1'\}$, $U=\{1,2\}$, $Y=\{1,2\}$.
	We already showed that system $\Sig_1$ is initial-state opaque in the proof of Proposition
	\ref{prop:SimulationOpacity_initial}, and
	system $\Sig_2$ is not initial-state opaque in Example \ref{rem4:initial_opaicty_NFTS}.
	Now consider relation $\sim=\{(1,1'),(2,2'),(3,3'),(4,4')\}$. One can verify that $\sim$
	satisfies all items of Definition \ref{def_OpacityPreservingsimulation_initial} other than
	\ref{item2c:opacity_initial}). Hence, \ref{item2c:opacity_initial}) in Definition \ref{def_OpacityPreservingsimulation_initial}
	is also necessary to make it initial-state opacity-preserving.

	\begin{figure}
        \centering
\begin{tikzpicture}[>=stealth',shorten >=1pt,auto,node distance=2.5 cm, scale = 0.9, transform shape,
	->,>=stealth,inner sep=2pt,state/.style={shape=circle,draw,top color=red!10,bottom color=blue!30},
	point/.style={circle,inner sep=0pt,minimum size=2pt,fill=}, 
	skip loop/.style={to path={-- ++(0,#1) -| (\tikztotarget)}}]
	\node[initial,state,accepting,initial by arrow,initial where=below] (1) {$1/1$};
	\node[initial,state,initial by arrow,initial where=below] (2) [right of = 1] {$2/2$};
	\node[initial,state,initial by arrow,initial where=below] (3) [right of = 2] {$3/1$};
	\node[initial,state,initial by arrow,initial where=below] (4) [right of = 3] {$4/2$};
	\node[initial,state,accepting,initial by arrow,initial where=above] (1') [below of = 1]  {$1'/1$};
	\node[initial,state,initial by arrow,initial where=above] (2') [right of = 1'] {$2'/2$};
	\node[initial,state,initial by arrow,initial where=above] (3') [right of = 2'] {$3'/1$};
	\node[initial,state,initial by arrow,initial where=above] (4') [right of = 3'] {$4'/2$};

	\path [->]
	(1') edge node {$1,2$} (2')
	(2') edge node {$1$} (3')
	(3') edge node {$1$} (4')
	(4') edge [bend left] node {$1$} (1')
	;

	\path [->]
	(1) edge node {$1$} (2)
	(2) edge node {$1$} (3)
	(3) edge node {$1$} (4)
	(4) edge [bend right] node {$1$} (1)
	;

	\node at (-1,0) {$\Sig_1$};
	\node at (-1,-2.5) {$\Sig_2$};
        \end{tikzpicture}
		\caption{State transition diagrams of two NFTSs in Example \ref{rem3:initial_opaicty_NFTS}.}
	\label{fig3:initial_opaicty_NFTS}
\end{figure}
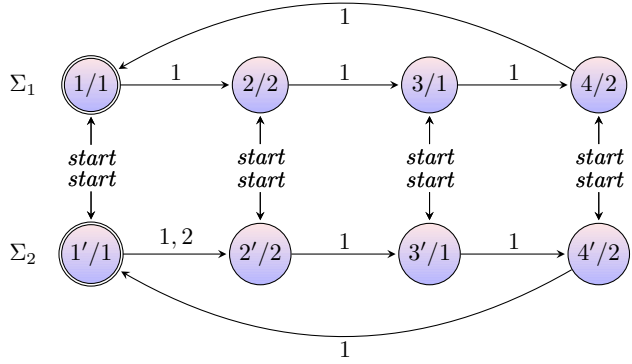

\end{example}

	One can conclude from Examples \ref{rem1:initial_opaicty_NFTS}, \ref{rem2:initial_opaicty_NFTS},
	\ref{rem3:initial_opaicty_NFTS}, and \ref{rem4:initial_opaicty_NFTS} that in order to make Definition
	\ref{def_OpacityPreservingsimulation_initial} initial-state opacity-preserving, all items
	\ref{item1a:opacity_initial}), \ref{item1c:opacity_initial}), \ref{item2a:opacity_initial}),
	and \ref{item2c:opacity_initial})
	are necessary. Therefore, the simulation relation introduced in Definition \ref{def_OpacityPreservingsimulation_initial} is
	a weak relation in terms of requiring minimal conditions preserving initial-state opacity of NTSs.

It is easy to see that Definition \ref{def_OpacityPreservingsimulation_initial} can only guarantee unidirectional
preservation of initial-state opacity.
Analogously, we can define an InitSOP bisimulation relation that ensures the bidirectional preservation
of initial-state opacity as in Definition \ref{def_OpacityPreservingBisimulation_initial}.
Definition \ref{def_OpacityPreservingBisimulation_initial} is a stronger version of bisimulation relation.

\begin{definition}[InitSOP bisimulation relation]\label{def_OpacityPreservingBisimulation_initial}
	Consider two NTSs $\Sig_i=(X_i,X_{i,0},S_i,U,\to_i,Y,h_i)$, $i=1,2$.
	A relation $\sim\subseteq X_1\times X_2$ is called an InitSOP bisimulation relation between $\Sig_1$ and $\Sig_2$  if
	\begin{enumerate}
		\item\label{item1:opacity_initial_bi}
			\begin{enumerate}
				\item\label{item1a:opacity_initial_bi}
					for all $x_{1,0}\in X_{1,0}\cap S_1$, there exists $x_{2,0}\in X_{2,0}\cap S_2$
					such that $(x_{1,0},x_{2,0})\in \sim$;
				\item\label{item1b:opacity_initial_bi}
					for all $x_{1,0}\in X_{1,0}\setminus S_1$, there exists $x_{2,0}\in X_{2,0}\setminus
					S_2$ such that $(x_{1,0},x_{2,0})\in \sim$;
				\item\label{item1c:opacity_initial_bi}
					for all $x_{2,0}\in X_{2,0}\cap  S_2$, there exists $x_{1,0}\in X_{1,0}\cap S_1$
					such that $(x_{1,0},x_{2,0})\in \sim$;
				\item\label{item1d:opacity_initial_bi}
					for all $x_{2,0}\in X_{2,0}\setminus S_2$, there exists $x_{1,0}\in X_{1,0}\setminus
					S_1$ such that $(x_{1,0},x_{2,0})\in \sim$;
			\end{enumerate}
		\item\label{item3:opacity_initial_bi}
			for every $(x_1,x_2)\in\sim$, $h_1(x_1)=h_2(x_2)$;	
		\item\label{item2:opacity_initial_bi}
			for every $(x_1,x_2)\in\sim$,
			\begin{enumerate}
				\item\label{item2a:opacity_initial_bi}
					for every transition $x_1\xrightarrow[]{u}_1 x_1'$, there exists transition
					$x_2\xrightarrow[]{u}_2 x_2'$ such that $(x_1',x_2')\in\sim$;
				\item\label{item2c:opacity_initial_bi}
					for every transition $x_2\xrightarrow[]{u}_2 x_2'$, there exists transition
					$x_1\xrightarrow[]{u}_1 x_1'$ such that $(x_1',x_2')\in\sim$;
			\end{enumerate}
	\end{enumerate}
\end{definition}

\begin{remark}
Consider two NTSs $\Sig_i=(X_i,X_{i,0},S_i,U,\to_i,Y,h_i)$, $i=1,2$. One can readily verify from Definition \ref{def_OpacityPreservingBisimulation_initial} that a relation $\sim\subseteq X_1\times X_2$ is called an InitSOP bisimulation relation between $\Sig_1$ and $\Sig_2$ if $\sim$ is an InitSOP simulation relation from $\Sig_1$ to $\Sig_2$ and\footnote{Given a relation $\sim\subseteq X_1\times X_2$, $\sim^{-1}$ denotes the inverse relation defined by $\sim^{-1}=\{(x_2,x_1)\in X_2\times X_1\,\,|\,\,(x_1,x_2)\in\sim\}$.} $\sim^{-1}$ is an InitSOP simulation relation from $\Sig_2$ to $\Sig_1$.
\end{remark}

Similar to Theorem \ref{thm1:opacityNFTS_initial}, the following theorem
follows from Definition \ref{def_OpacityPreservingBisimulation_initial}.

\begin{theorem}\label{thm1:opacityNFTS_initial-2side}
	Consider two NTSs $\Sig_i=(X_i,X_{i,0},S_i,U,\to_i,Y,h_i)$, $i=1,2$. Assume that there exists
	an InitSOP bisimulation relation $\sim\subseteq X_1\times X_2$ between $\Sig_1$ and $\Sig_2$.
	Then $\Sig_1$ is initial-state opaque if and only if
	$\Sig_2$ is also initial-state opaque.
\end{theorem}

\begin{proof}
Since $\Sig_1$ simulates $\Sig_2$ and vice versa as in Definition \ref{def_OpacityPreservingsimulation_initial}, the proof is a simple consequence of the proof of Theorem \ref{thm1:opacityNFTS_initial}.
\end{proof}

\subsection{Initial-state opacity-preserving quotient relation}

From the results in the previous subsection, one can verify initial-state opacity
of system $\Sig_2$ by verifying it over system $\Sig_1$ (resp. verify lack of initial-state opacity
of system $\Sig_1$ by verifying it over system $\Sig_2$)
provided that there exists an InitSOP simulation relation from $\Sig_1$ to $\Sig_2$.
In this subsection, we show that the quotient relation defined in \eqref{eqn:quotientrelation}
from an NTS to its quotient system
is an InitSOP bisimulation relation under certain mild assumptions.
Hence, one can leverage the existing bisimulation algorithms provided in
\cite{tab09} with some modifications to construct InitSOP abstractions (if existing).

\begin{theorem}\label{thm2:opacityNFTS_initial}
	Let $\Sig=(X,X_0,S,U,\to,Y,h)$ be an NTS and $\sim\subseteq X\times X$ be an equivalence relation on $X$
	satisfying $h(x)=h(x')$ for all $(x,x')\in\sim$.
	Assume that for all $x\in S$ and $x'\in X$, if $(x,x')\in\sim$
	then $x'\in S$.
	Then $\sim_{\mathsf{Q}}$ is an InitSOP bisimulation relation
	between $\Sig$ and $\Sig_{\sim}$ if and only if relation $\sim$ satisfies
	\begin{equation}
	\forall(x,x')\in \sim,\forall~x\xrightarrow[]{u}x'',\exists~x'\xrightarrow[]{u}x''' \text{ with }(x'',x''')\in\sim.
		\label{eqn1:opacity}
	\end{equation}
\end{theorem}

\begin{proof}
	By assumption, for all $x\in S$ and $x'\in X$, if $(x,x')\in\sim$
	then $x'\in S$. This is equivalent to saying that for all $x\in X$, either $[x]\subseteq S$ or $[x]\cap S=\emptyset$,
	i.e., $S_{\sim}=\{[x]|x\in S\}$.

	(if:) Assume $\sim$ satisfies \eqref{eqn1:opacity}. Then by assumption,
	we next prove that $\sim_\mathsf{Q}$ is an InitSOP bisimulation relation between $\Sig$ and $\Sig_\sim$.

	For each $x_0\in X_0\cap S$, we have $[x_0]\in X_{\sim,0}\cap S_\sim$, i.e.,
	\ref{item1a:opacity_initial_bi}) of Definition \ref{def_OpacityPreservingBisimulation_initial} holds.
	Similarly \ref{item1b:opacity_initial_bi}) of Definition \ref{def_OpacityPreservingBisimulation_initial} holds.

	For each $[x_0]\in X_{\sim,0}\cap S_\sim$, there exists $x'\in X_0$ satisfying
	$x_0\sim x'$. By assumption, we also have $x_0\in S$, and then $x'\in S$, i.e.,
	\ref{item1c:opacity_initial_bi}) of Definition \ref{def_OpacityPreservingBisimulation_initial} holds.
	Similarly \ref{item1d:opacity_initial_bi}) of Definition \ref{def_OpacityPreservingBisimulation_initial} holds.

	Condition \ref{item3:opacity_initial_bi}) in Definition \ref{def_OpacityPreservingBisimulation_initial} naturally holds
	by the assumption.

	For each $x\in X$, we have $(x,[x])\in\sim_\mathsf{Q}$.

	If there exists transition $x\xrightarrow[]{u}x'$ in $\Sig$, then there exists transition
	$[x]\xrightarrow[]{u}_\sim [x']$ in $\Sig_\sim$, and $(x',[x'])\in\sim_{\mathsf{Q}}$,
	i.e., \ref{item2a:opacity_initial_bi}) in Definition \ref{def_OpacityPreservingBisimulation_initial} holds.

	If there exists transition $[x]\xrightarrow[]{u}_\sim [x']$ in $\Sig_\sim$, then there exists transition
	$x''\xrightarrow[]{u}x'''$ in $\Sig$ satisfying that $x\sim x''$ and $x'\sim x'''$.
	By \eqref{eqn1:opacity}, there exists transition $x\xrightarrow[]{u}x''''$ such that $x'''\sim x''''$,
	then $x'\sim x''''$ and $(x'''',[x'])\in\sim_\mathsf{Q}$, i.e., \ref{item2c:opacity_initial_bi}) in Definition
	\ref{def_OpacityPreservingBisimulation_initial} holds, which completes the ``if'' part.

	(only if:) Assume that $\sim_{\mathsf{Q}}$ is an InitSOP bisimulation relation between $\Sig$ and $\Sig_\sim$.
	Next we prove \eqref{eqn1:opacity} holds. For each $(x,x')\in\sim$ and each transition
	$x\xrightarrow[]{u}x''$ in $\Sig$, we have $(x,[x'])\in \sim_{\mathsf{Q}}$,
	and by assumption, there exists transition $[x']\xrightarrow[]{u}_\sim[x''']$ in $\Sig_\sim$
	and $(x'',[x'''])\in\sim_\mathsf{Q}$. Then $(x'',x''')\in\sim$, i.e., \eqref{eqn1:opacity} holds.
\end{proof}

\subsection{Infinite-step opacity-preserving bisimulation relation}\label{subsec:InfSOP_Simulation}

We have given InitSOP (bi)simulation relation. Next we study whether InitSOP (bi)simulation relation
preserves the other three types of opacity; and if not, we propose new (bi)simulation relations that
preserve the other three types of opacity.

Similar to initial-state opacity, the classical bisimulation relation does not preserve the other three types of opacity.
See the NFTSs shown in the proof of Proposition \ref{prop:SimulationOpacity_initial} (cf. Fig.
\ref{fig1:initial_opaicty_NFTS}). One can easily verify that $\Sig_2$ in Fig. \ref{fig1:initial_opaicty_NFTS}
is not current-state opaque, or $K$-step opaque for any positive integer $K$, or infinite-step opaque.
However, $\Sig_1$ in Fig. \ref{fig1:initial_opaicty_NFTS} is current-state opaque, $K$-step opaque for any
positive integer $K$, and infinite-step opaque. In addition, under the relation $\sim=\{(1,1'),(2,2'),(1,3'),(2,4')\}$,
$\Sig_2$ bisimulates $\Sig_1$. Hence the following result holds.

\begin{proposition}\label{prop:SimulationOpacity}
	Bisimulation relation (cf. Definition \ref{def_bisimulation}) does not preserve current-state opacity,
	$K$-step opacity, or infinite-step opacity.
\end{proposition}

Since all these three types of opacity require that the intruder cannot make sure whether the current state
is secret, the previous InitSOP (bi)simulation relation does not suffice to preserve them either.
In this subsection, we strengthen the InitSOP bisimulation relation to make it preserve these three types of opacity.

\begin{definition}[InfSOP bisimulation relation]\label{def_OpacityPreservingBisimulation}
	Consider two NTSs $\Sig_i=(X_i,X_{i,0},S_i,U,\to_i,Y,h_i)$, $i=1,2$.
	A relation $\sim\subseteq X_1\times X_2$ is called an
	InfSOP bisimulation relation between $\Sig_1$ and $\Sig_2$  if
	\begin{enumerate}
		\item\label{item1:opacity}
			\begin{enumerate}
				\item\label{item1a:opacity}
					for all $x_{1,0}\in S_1\cap X_{1,0}$, there exists $x_{2,0}\in S_2\cap X_{2,0}$
					such that $(x_{1,0},x_{2,0})\in \sim$;
				\item\label{item1b:opacity}
					for all $x_{1,0}\in X_{1,0}\setminus S_1$, there exists $x_{2,0}\in X_{2,0}\setminus
					S_2$ such that $(x_{1,0},x_{2,0})\in \sim$;
				\item\label{item1c:opacity}
					for all $x_{2,0}\in S_2\cap X_{2,0}$, there exists $x_{1,0}\in S_1\cap X_{1,0}$
					such that $(x_{1,0},x_{2,0})\in \sim$;
				\item\label{item1d:opacity}
					for all $x_{2,0}\in X_{2,0}\setminus S_2$, there exists $x_{1,0}\in X_{1,0}\setminus
					S_1$ such that $(x_{1,0},x_{2,0})\in \sim$;
			\end{enumerate}
		\item\label{item3:opacity} for every $(x_1,x_2)\in\sim$, $h_1(x_1)=h_2(x_2)$;	
		\item\label{item2:opacity} for every $(x_1,x_2)\in\sim$,
			\begin{enumerate}
				\item\label{item2a:opacity}
					for every transition $x_1\xrightarrow[]{u}_1 x_1'\in S_1$, there exists transition
					$x_2\xrightarrow[]{u}_2 x_2'\in S_2$ such that $(x_1',x_2')\in\sim$;
				\item\label{item2b:opacity}
					for every transition $x_1\xrightarrow[]{u}_1 x_1'\in X_1\setminus S_1$, there exists transition
					$x_2\xrightarrow[]{u}_2 x_2'\in X_2\setminus S_2$ such that $(x_1',x_2')\in\sim$;
				\item\label{item2c:opacity}
					for every transition $x_2\xrightarrow[]{u}_2 x_2'\in S_2$, there exists transition
					$x_1\xrightarrow[]{u}_1 x_1'\in S_1$ such that $(x_1',x_2')\in\sim$;
				\item\label{item2d:opacity}
					for every transition $x_2\xrightarrow[]{u}_2 x_2'\in X_2\setminus S_2$, there exists transition
					$x_1\xrightarrow[]{u}_1 x_1'\in X_1\setminus S_1$ such that $(x_1',x_2')\in\sim$.
			\end{enumerate}
	\end{enumerate}
\end{definition}

%

Intuitively, condition \ref{item1:opacity}) ensures that each initial secret (non-secret) state in
$\Sig_1$ has a corresponding initial secret (non-secret) state in $\Sig_2$ such that they are in the relation,
and vice versa; condition
\ref{item2:opacity}) guarantees that each transition to a secret (non-secret) state in $\Sig_1$
has a corresponding transition to a secret (non-secret) state in $\Sig_2$, and vice versa. Conditions
\ref{item1:opacity}) and \ref{item2:opacity}) make bisimulation relation preserve infinite-step opacity,
which is shown in the following theorem.

\begin{theorem}\label{thm1:opacityNFTS}
	Consider two NTSs $\Sig_i=(X_i,X_{i,0},S_i,U,\to_i,Y,h_i)$, $i=1,2$.
	If there exists an InfSOP bisimulation relation $\sim\subseteq X_1\times X_2$ between $\Sig_1$ and $\Sig_2$,
	then $\Sig_1$ is infinite-step opaque if and only if $\Sig_2$ is infinite-step opaque.
\end{theorem}

\begin{proof}
	Assume there exists an InfSOP bisimulation relation $\sim\subseteq X_1\times X_2$ between $\Sig_1$ and $\Sig_2$
	and system $\Sig_1$ is infinite-step opaque. Now we show that $\Sig_2$ is also infinite-step opaque.

	For system $\Sig_2$, we arbitrarily choose input sequence $\a\in U^*$, states $x_{2,0}\in X_{2,0}$
	and $x_{2,1},\dots,x_{2,|\a|}\in X_2$ such that
	$$x_{2,0}\xrightarrow[]{\a(0)}_2 x_{2,1}\xrightarrow[]{\a(1)}_2\cdots\xrightarrow[]{\a(|\a|-1)}_2
	x_{2,|\a|},$$ and $x_{2,i}\in S_2$ for any $i\in[0,|\a|]$.

	Since $\Sig_1$ simulates $\Sig_2$, by \ref{item1c:opacity}), \ref{item1d:opacity}), \ref{item3:opacity}), \ref{item2c:opacity}),
	and \ref{item2d:opacity}), there exist $x_{1,0}\in X_{1,0}$, $x_{1,j}\in X_1$, $j\in[1,|\a|]$
	such that $x_{1,k}\in S_1$, $h_1(x_{1,k})=h_2(x_{2,k})$, $k\in[0,|\a|]$, and
	$$x_{1,0}\xrightarrow[]{\a(0)}_1 x_{1,1}\xrightarrow[]{\a(1)}_1\cdots\xrightarrow[]{\a(|\a|-1)}_1
	x_{1,|\a|}.$$

	Since $\Sig_1$ is infinite-step opaque, there exist $x_{1,0}'\in X_{1,0}$, $x_{1,j}'\in X_1$,
	$j\in[1,|\a|]$ such that $x_{1,k}'\in X_1\setminus S_1$, $h_1(x_{1,k})=h_1(x_{1,k}')$, $k\in[0,|\a|]$, and
	$$x_{1,0}'\xrightarrow[]{\a(0)}_1 x_{1,1}'\xrightarrow[]{\a(1)}_1\cdots\xrightarrow[]{\a(|\a|-1)}_1
	x_{1,|\a|}'.$$

	Since $\Sig_2$ simulates $\Sig_1$, by \ref{item1a:opacity}), \ref{item1b:opacity}), \ref{item3:opacity}), \ref{item2a:opacity}),
	and \ref{item2b:opacity}), there exist $x_{2,0}'\in X_{2,0}$
	and $x_{2,1}',\dots,x_{2,|\a|}'\in X_2$ such that $x_{2,j}'\in X_2\setminus S_2$,
	$h_1(x_{1,j}')=h_2(x_{2,j}')$, $j\in[0,|\a|]$, and
	$$x_{2,0}'\xrightarrow[]{\a(0)}_2 x_{2,1}'\xrightarrow[]{\a(1)}_2\cdots\xrightarrow[]{\a(|\a|-1)}_2
	x_{2,|\a|}'.$$ Hence $h_2(x_{2,j})=h_2(x_{2,j}')$, $j\in[0,|\a|]$, and $\Sig_2$ is infinite-step opaque.

	Symmetrically, assume that there exists an InfSOP bisimulation relation $\sim\subseteq X_1\times X_2$
	between $\Sig_1$ and $\Sig_2$
	and system $\Sig_2$ is infinite-step opaque,
	we can show that $\Sig_1$ is also infinite-step opaque.
\end{proof}

By the similarity of Definitions \ref{def_current-state_opacity}, \ref{def_K-step_opacity},
and \ref{def_infinite-step_opacity}, the following corollary follows.

\begin{corollary}
	Consider two NTSs $\Sig_i=(X_i,X_{i,0},S_i,U,\to_i,Y,h_i)$, $i=1,2$. If there exists
	an InfSOP bisimulation relation $\sim\subseteq X_1\times X_2$ between $\Sig_1$ and $\Sig_2$,
	then $\Sig_1$ is current-state (resp. $K$-step) opaque
	if and only if $\Sig_2$ is current-state (resp. $K$-step) opaque.
\end{corollary}

\begin{remark}
Note that although we add several additional conditions in Definition \ref{def_OpacityPreservingBisimulation}
to make bisimulation relation preserving these three types of opacity, these conditions are somehow necessary.
That is, without some of them, bisimulation relation may not preserve those notions of opacity any more.
Taking the two NFTSs shown in Fig. \ref{fig1:initial_opaicty_NFTS} for example, bisimulation relation
$\sim=\{(1',1),(2',2),(3',1),(4',2)\}$ satisfies \ref{item1a:opacity}),
\ref{item1c:opacity}), \ref{item1d:opacity}), \ref{item3:opacity}), \ref{item2a:opacity}), and \ref{item2d:opacity}),
but does not satisfy \ref{item1b:opacity}), \ref{item2b:opacity}), or \ref{item2c:opacity}).
\end{remark}

\begin{remark}
Note that since the preservation of infinite-step opacity always requires a bidirectional relation, so we directly study InfSOP bisimulation relation.
A detailed study of relevant notions of (bi)simulation relation for preserving current-state and $K$-step opacity are left for future investigations.
\end{remark}

\subsection{Infinite-step opacity-preserving quotient relation}\label{subsec:InfSOP_Quotient}
In this subsection, we again use the quotient relation from an NTS to its quotient system
to implement the InfSOP bisimulation relation.

\begin{theorem}\label{thm2:opacityNFTS}
	Let $\Sig=(X,X_0,S,U,\to,Y,h)$ be an NTS and $\sim\subseteq X\times X$ be an equivalence relation on $X$
	satisfying $h(x)=h(x')$ for all $(x,x')\in\sim$.
	Assume that
	for all $x\in S$ and $x'\in X$, if $(x,x')\in\sim$
	then $x'\in S$.
	Then $\sim_{\mathsf{Q}}$ is an InfSOP bisimulation relation
	between $\Sig$ and $\Sig_{\sim}$ if and only if $\sim$ is an InfSOP bisimulation relation
	between $\Sig$ and itself.
\end{theorem}

\begin{proof}
	Similar to Theorem \ref{thm2:opacityNFTS_initial}, by assumption we have
	for all $x\in X$, either $[x]\subseteq S$ or $[x]\cap S=\emptyset$,
	i.e., $S_{\sim}=\{[x]|x\in S\}$.

	(if:) Assume that $\sim$ is an InfSOP bisimulation relation between $\Sig$ and itself. Next we show that $\sim_{\mathsf{Q}}$ is also
	an InfSOP bisimulation relation between $\Sig$ and $\Sig_{\sim}$ according to Definition \ref{def_OpacityPreservingBisimulation}.

	For all $x\in X_0\cap S$, we have $[x]\in X_{\sim,0}\cap S_{\sim}$, and $(x,[x])\in\sim_{\mathsf{Q}}$,
	i.e., \ref{item1a:opacity}) in Definition \ref{def_OpacityPreservingBisimulation} holds.

	For all  $x\in X_0\setminus S$, we have $[x]\in X_{\sim,0}\setminus S_{\sim}$, and $(x,[x])\in\sim_{\mathsf{Q}}$,
	i.e., \ref{item1b:opacity}) in Definition \ref{def_OpacityPreservingBisimulation} holds.

	For all $[x]\in X_{\sim,0}\cap S_{\sim}$, we have $[x]\cap X_0\ne\emptyset$, and $[x]\subseteq S$,
	then there exists $\bar x\in[x]$ such that $\bar x\in X_0\cap S$, and $(\bar x,[x])\in\sim_{\mathsf{Q}}$,
	i.e.,  \ref{item1c:opacity}) in Definition \ref{def_OpacityPreservingBisimulation} holds.
	
	For all $[x]\in X_{\sim,0}\setminus S_{\sim}$, we have $[x]\cap X_0\ne\emptyset$, and $[x]\cap S=\emptyset$,
	then there exists $\bar x\in[x]$ such that $\bar x\in X_0\setminus S$, and $(\bar x,[x])\in\sim_{\mathsf{Q}}$,
	i.e.,  \ref{item1d:opacity}) in Definition \ref{def_OpacityPreservingBisimulation} holds.

	Now consider an arbitrary pair $(\bar x,[x])\in\sim_{\mathsf{Q}}$, i.e., $[\bar x]=[x]$,
	$\bar x\sim x$. By definition we have $h(\bar x) =h (x)=h_\sim([x])$.
	
	Now consider an arbitrary pair $(\bar x,[x])\in\sim_{\mathsf{Q}}$, i.e., $\bar x\in[x]$.

	For every transition $\bar x\xrightarrow[]{u}\bar x'\in S$, where $u\in U$, we have $[x]\xrightarrow[]{u}_{\sim}
	[\bar x']\in S_{\sim}$, and $(\bar x',[\bar x'])\in\sim_{\mathsf{Q}}$,
	i.e., \ref{item2a:opacity}) in Definition \ref{def_OpacityPreservingBisimulation} holds.

	For every transition $\bar x\xrightarrow[]{u}\bar x'\in X\setminus S$, where $u\in U$,
	we have $[x]\xrightarrow[]{u}_{\sim}
	[\bar x']\in X_{\sim}\setminus S_{\sim}$, and $(\bar x',[\bar x'])\in\sim_{\mathsf{Q}}$,
	i.e., \ref{item2b:opacity}) in Definition \ref{def_OpacityPreservingBisimulation} holds.

	For every transition $[x]\xrightarrow[]{u}_{\sim}[x']\in S_{\sim}$, where $u\in U$, there exists transition
	$\hat x\xrightarrow[]{u}\hat x'\in S$ such that $\hat x\in[x]$, and $\hat x'\in[x']$.
	Since $(\bar x,\hat x)\in\sim$, and $\sim$ is InfSOP, there exists transition
	$\bar x\xrightarrow[]{u}\bar x'\in S$ such that $(\hat x',\bar x')\in \sim$, hence $(\bar x',[x'])
	\in\sim_{\mathsf{Q}}$, i.e., \ref{item2c:opacity}) in Definition \ref{def_OpacityPreservingBisimulation} holds.

	For every transition $[x]\xrightarrow[]{u}_{\sim}[x']\in X_{\sim}\setminus S_{\sim}$,
	where $u\in U$, there exists transition
	$\hat x\xrightarrow[]{u}\hat x'\in X\setminus S$ such that $\hat x\in[x]$, and $\hat x'\in[x']$.
	Since $(\bar x,\hat x)\in\sim$, and $\sim$ is InfSOP, there exists transition
	$\bar x\xrightarrow[]{u}\bar x'\in X\setminus S$ such that $(\hat x',\bar x')\in \sim$, hence $(\bar x',[x'])
	\in\sim_{\mathsf{Q}}$, i.e., \ref{item2d:opacity}) in Definition \ref{def_OpacityPreservingBisimulation} holds.
	Hence $\sim_{\mathsf{Q}}$ is InfSOP.

	(only if:) Assume that $\sim_{\mathsf{Q}}$ is an InfSOP bisimulation relation between $\Sig$ and $\Sig_{\sim}$.
	Now we show that $\sim$ is also
	an InfSOP bisimulation relation between $\Sig$ and itself according to Definition \ref{def_OpacityPreservingBisimulation}.
	Since $\sim$ is an equivalence relation, we have $(x,x)\in\sim$ for all $x\in X$.

	For all $x\in X_0\cap S$, we have $(x,x)\in\sim$,
	i.e., \ref{item1a:opacity}) in Definition \ref{def_OpacityPreservingBisimulation} holds.
	Similarly, \ref{item1b:opacity}), \ref{item1c:opacity}), and \ref{item1d:opacity})
	in Definition \ref{def_OpacityPreservingBisimulation} hold.

	By the definition of $\sim$, we have $h(x_1)=h(x_2)$ for all $(x_1,x_2)\in\sim$.
	Hence \ref{item3:opacity}) in Definition \ref{def_OpacityPreservingBisimulation} holds.

	
	Now consider an arbitrary pair $(x_1,x_2)\in\sim$.

	For every transition $x_1\xrightarrow[]{u}x_1'\in S$, where $u\in U$, we have $[x_1]\xrightarrow[]{u}_{\sim}
	[x_1']\in S_{\sim}$. Since $\sim_{\mathsf{Q}}$ is InfSOP, and $(x_2,[x_1])\in\sim_{\mathsf{Q}}$,
	there exists transition $x_2\xrightarrow[]{u}x_2'\in S\cap[x_1']=[x_1']$, then $(x_1',x_2')\in\sim$,
	i.e., \ref{item2a:opacity}) in Definition \ref{def_OpacityPreservingBisimulation} holds.

	For every transition $x_1\xrightarrow[]{u}x_1'\in X\setminus S$,
	where $u\in U$, we have $[x_1]\xrightarrow[]{u}_{\sim}
	[x_1']\in X_{\sim}\setminus S_{\sim}$. Since $\sim_{\mathsf{Q}}$ is InfSOP,
	and $(x_2,[x_1])\in\sim_{\mathsf{Q}}$,
	there exists transition $x_2\xrightarrow[]{u}x_2'\in (X\setminus S)\cap[x_1']=[x_1']$,
	then $(x_1',x_2')\in\sim$,
	i.e., \ref{item2b:opacity}) in Definition \ref{def_OpacityPreservingBisimulation} holds.

	Symmetrically, \ref{item2c:opacity}) and \ref{item2d:opacity})
	in Definition \ref{def_OpacityPreservingBisimulation} hold. Hence, $\sim$ is an InfSOP
	bisimulation relation between $\Sig$ and itself.
\end{proof}


\begin{example}\label{exam4_OpacityNFTS}
	Consider NFTS $\Sig=(X,X_{0},S,U,\to,Y,h)$
	shown in Fig. \ref{fig4:opaicty_NFTS}, where $X=\{1,2,3,4,5,6,7,8\}=X_0$, $S=\{1,5\}$, $U=\{1\}$,
	$Y=\{1,2\}$.
	It can be readily seen that the equivalence relation $\sim=\{(1,1),(2,2),(3,3),\\(4,4),(5,5),(6,6),(7,7),(8,8),(1,5),(5,1),(2,6),
	(6,2),(3,7),(7,3),(4,8),(8,4)\}\subseteq X\times X$ is an InfSOP bisimulation relation between $\Sig$ and itself.
	Under this relation, the quotient system of $\Sig$ is $\Sig_{\sim}=(X_{\sim},X_{\sim,0},S_\sim,U,\to_\sim,
	Y,h_\sim)$, where $X_{\sim}=X/\sim=X_{\sim,0}$, $X/\sim=\{\{1,5\},\{2,6\},\{3,7\},\{4,8\}\}$, $S_{\sim}=\{\{1,5\}\}$, which is shown in
	Fig. \ref{fig5:opaicty_NFTS}. It can be easily seen that $\Sig_{\sim}$ is infinite-step opaque. Therefore,
	the original NFTS $\Sig$ is also infinite-step opaque due to the results in Theorem \ref{thm2:opacityNFTS}.

	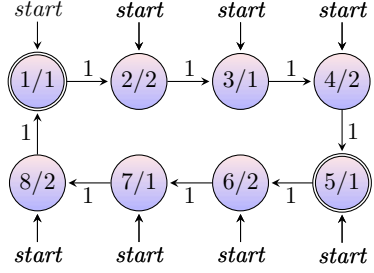
\begin{figure}
        \centering
\begin{tikzpicture}[>=stealth',shorten >=1pt,auto,node distance=1.5 cm, scale = 0.9, transform shape,
	->,>=stealth,inner sep=2pt,state/.style={shape=circle,draw,top color=red!10,bottom color=blue!30},
	point/.style={circle,inner sep=0pt,minimum size=2pt,fill=}, 
	skip loop/.style={to path={-- ++(0,#1) -| (\tikztotarget)}}]
	\node[state,accepting,initial by arrow,initial where=above] (1) {$1/1$};
	\node[initial,state,initial by arrow,initial where=above] (2) [right of = 1] {$2/2$};
	\node[initial,state,initial by arrow,initial where=above] (3) [right of = 2] {$3/1$};
	\node[initial,state,initial by arrow,initial where=above] (4) [right of = 3] {$4/2$};
	\node[initial,accepting,state,initial by arrow,initial where=below] (5) [below of = 4] {$5/1$};
	\node[initial,state,initial by arrow,initial where=below] (6) [left of = 5] {$6/2$};
	\node[initial,state,initial by arrow,initial where=below] (7) [left of = 6] {$7/1$};
	\node[initial,state,initial by arrow,initial where=below] (8) [left of = 7] {$8/2$};

	\path [->]
	(1) edge node {$1$} (2)
	(2) edge node {$1$} (3)
	(3) edge node {$1$} (4)
	(4) edge node {$1$} (5)
	(5) edge node {$1$} (6)
	(6) edge node {$1$} (7)
	(7) edge node {$1$} (8)
	(8) edge node {$1$} (1)
	;

        \end{tikzpicture}
		\caption{State transition diagram of the NFTS in Example \ref{exam4_OpacityNFTS}.}
	\label{fig4:opaicty_NFTS}
\end{figure}

	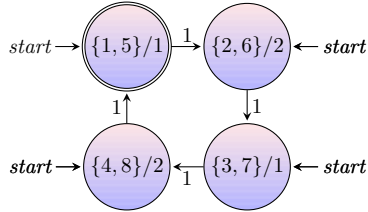
\begin{figure}
        \centering
\begin{tikzpicture}[>=stealth',shorten >=1pt,auto,node distance=2  cm, scale = 0.8, transform shape,
	->,>=stealth,inner sep=2pt,state/.style={shape=circle,draw,top color=red!10,bottom color=blue!30},
	point/.style={circle,inner sep=0pt,minimum size=2pt,fill=}, 
	skip loop/.style={to path={-- ++(0,#1) -| (\tikztotarget)}}]
	\node[state,accepting,initial by arrow,initial where=left] (1) {$\{1,5\}/1$};
	\node[initial,state,initial by arrow,initial where=right] (2) [right of = 1] {$\{2,6\}/2$};
	\node[initial,state,initial by arrow,initial where=right] (7) [below of = 2] {$\{3,7\}/1$};
	\node[initial,state,initial by arrow,initial where=left] (8) [left of = 7] {$\{4,8\}/2$};

	\path [->]
	(1) edge node {$1$} (2)
	(2) edge node {$1$} (7)
	(7) edge node {$1$} (8)
	(8) edge node {$1$} (1)
	;

        \end{tikzpicture}
		\caption{State transition diagram of the quotient system of the NFTS in Example \ref{exam4_OpacityNFTS}
		shown in Fig. \ref{fig4:opaicty_NFTS}.}
	\label{fig5:opaicty_NFTS}
\end{figure}

\end{example}

\section{Relationship Between Different Notions of Opacity}\label{sec:opacity_relation}

In this section, we characterize the relationship between different notions of opacity for
NTSs.

\begin{theorem}
	The implication relationship between different notions of opacity for NTSs is shown in Fig.
	\ref{fig6:opaicty_NFTS}.
	\label{thm3:opacityNFTS}
\end{theorem}

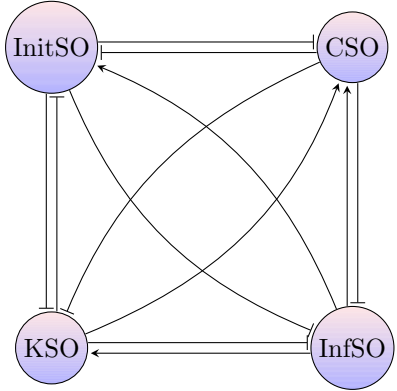
\begin{figure}
        \centering
\begin{tikzpicture} [>=stealth',shorten >=1pt,auto,node distance=4.0 cm, scale = 1.0, transform shape,
	->,>=stealth,inner sep=2pt,state/.style={shape=circle,draw,top color=red!10,bottom color=blue!30},
	point/.style={circle,inner sep=0pt,minimum size=2pt,fill=},
	skip loop/.style={to path={-- ++(0,#1) -| (\tikztotarget)}}]
	\node[state] (7) {InitSO};
	\node[state] [right of =7] (4) {CSO};
	\node[state] [below of =4] (5) {InfSO};
	\node[state] [left of =5] (8) {KSO};
	\draw [-|] ([yshift=2pt] 7.east) -- ([yshift=2pt] 4.west);
	\draw [-|] ([yshift=-2pt] 4.west) -- ([yshift=-2pt] 7.east);
	\draw [-|] ([xshift=2pt] 4.south) -- ([xshift=2pt] 5.north);
	\draw [->] ([xshift=-2pt] 5.north) -- ([xshift=-2pt] 4.south);
	\draw [-|] ([yshift=2pt] 8.east) -- ([yshift=2pt] 5.west);
	\draw [->] ([yshift=-2pt] 5.west) -- ([yshift=-2pt] 8.east);
	\draw [->] (8) to [out=22.5, in=249.5] (4);
	\draw [-|] (4) to [out=204.5, in=67.5] (8);
	\draw [-|] ([xshift=2pt] 8.north) -- ([xshift=2pt] 7.south);
	\draw [-|] ([xshift=-2pt] 7.south) -- ([xshift=-2pt] 8.north);
	\draw [->] (5) to [out=112.5, in=-22.5] (7);
	\draw [-|] (7) to [out=-67.5, in=157.5] (5);
        \end{tikzpicture}
	\caption{Implication relationship between different notions of opacity,
		where each pointed arrow means ``implies'' and each blunt arrow means ``does not imply''.
	}
	\label{fig6:opaicty_NFTS}
\end{figure}

\begin{proof}
	By Definitions
	\ref{def_initial-state_opacity}, \ref{def_current-state_opacity},
	\ref{def_K-step_opacity}, and \ref{def_infinite-step_opacity},
	one directly sees that InfSO implies KSO, CSO, and InitSO; and KSO implies CSO.
	We use counterexamples to prove the remaining parts as in Fig. \ref{fig6:opaicty_NFTS}.

	First, consider the NFTS as in Fig. \ref{fig7:opaicty_NFTS}, where $X=\{x_1,x_2,x_3,x_4\}=X_0$,
	$S=\{x_1\}$, $U=\{u_1,u_2\}$, and $Y=\{y_1,y_2,y_3\}$.
	For any input sequence $u_1^*\in U^*$, where $u_1^*$ means an arbitrary finite
	sequence consisting of $u_1$'s and including $\e$, the unique run over input sequence $u_1^*$
	ending at the unique secret state $x_1$ is $x_3^*x_1$,
	and there exists another run $x_3^*x_2$ also over the same input sequence
	producing the same output sequence as run $x_3^*x_1$, where state $x_2$ is not secret.
	For any input sequence containing input $u_2$,
	there exists no run over which one can end at secret state $x_1$. Hence, the NFTS is current-state opaque.
	Consider input sequence $u_2$ and run $x_1x_4$ over $u_2$. There is no other run also over $u_2$,
	so the NFTS is not initial-state opaque. Hence, CSO does not imply InitSO. Consider an arbitrary input sequence
	$u_1^*u_2u_1^*\in U^*$ and run $x_3^*x_1x_4^*$ over it. There is no other run over it, hence,
	the NFTS is not $K$-step opaque for any positive integer $K$. Hence CSO does not imply KSO for any
	positive integer $K$.
	\begin{figure}
        \centering
\begin{tikzpicture}[>=stealth',shorten >=1pt,auto,node distance=2.5 cm, scale = 0.8, transform shape,
	->,>=stealth,inner sep=2pt,state/.style={shape=circle,draw,top color=red!10,bottom color=blue!30},
	point/.style={circle,inner sep=0pt,minimum size=2pt,fill=}, 
	skip loop/.style={to path={-- ++(0,#1) -| (\tikztotarget)}}]
	\node[initial,state,accepting,initial by arrow,initial where=left] (1) {$x_1/y_1$};
	\node[initial,state,initial by arrow,initial where=above] (2) [below of = 1] {$x_2/y_1$};
	\node[initial,state,initial by arrow,initial where=above] (3) [left of = 2] {$x_3/y_2$};
	\node[initial,state,initial by arrow,initial where=above] (4) [right of = 2] {$x_4/y_3$};

	\path [->]
	(3) edge node {$u_1$} (1)
	(3) edge node {$u_1$} (2)
	(3) edge [loop left] node {$u_1$} (3)
	(1) edge node {$u_1,u_2$} (4)
	(2) edge node {$u_1$} (4)
	(4) edge [loop right] node {$u_1$} (4)
	;

        \end{tikzpicture}
		\caption{State transition diagram of an NFTS in the proof of Theorem \ref{thm3:opacityNFTS}.}
		\label{fig7:opaicty_NFTS}
\end{figure}

	Second, consider the NFTS as in Fig. \ref{fig8:opaicty_NFTS}, where $X=\{x_1,x_2,x_3,x_4,x_5\}=X_0$,
	$S=\{x_1\}$, $U=\{u_1,u_2\}$, and $Y=\{y_1,y_2,y_3\}$.
	For any input sequence $u_1^*u_1u_1\in U^*$ and any run $x_3^*x_1x_4$ over $u_1^*u_1u_1$,
	there is run $x_3^*x_2x_5$ also over the same input sequence and producing the same output sequence
	as $x_3^*x_1x_4$. There exists no other run such that the unique secret state $x_1$ is at the last but one
	time step. So the NFTS is $1$-step opaque. Similarly one sees that the NFTS is not $K$-step opaque for
	any integer $K>1$, hence, not infinite-step opaque. Therefore, KSO does not imply InfSO.
	\begin{figure}
        \centering
\begin{tikzpicture}[>=stealth',shorten >=1pt,auto,node distance=2.8 cm, scale = 0.8, transform shape,
	->,>=stealth,inner sep=2pt,state/.style={shape=circle,draw,top color=red!10,bottom color=blue!30},
	point/.style={circle,inner sep=0pt,minimum size=2pt,fill=}, 
	skip loop/.style={to path={-- ++(0,#1) -| (\tikztotarget)}}]
	\node[initial,state,accepting,initial by arrow,initial where=below] (1) {$x_1/y_1$};
	\node[initial,state,initial by arrow,initial where=above] (2) [below of = 1] {$x_2/y_1$};
	\node[initial,state,initial by arrow,initial where=above] (3) [left of = 2] {$x_3/y_2$};
	\node[initial,state,initial by arrow,initial where=below] (4) [right of = 1] {$x_4/y_3$};
	\node[initial,state,initial by arrow,initial where=above] (5) [right of = 2] {$x_5/y_3$};

	\path [->]
	(3) edge node {$u_1$} (1)
	(3) edge node {$u_1$} (2)
	(3) edge [loop left] node {$u_1$} (3)
	(1) edge node {$u_1$} (4)
	(2) edge node {$u_1$} (5)
	(4) edge [loop right] node {$u_2$} (4)
	(5) edge [loop right] node {$u_1$} (5)
	;

        \end{tikzpicture}
		\caption{State transition diagram of an NFTS in the proof of Theorem \ref{thm3:opacityNFTS}.}
		\label{fig8:opaicty_NFTS}
\end{figure}

	Third, consider the NFTS as in Fig. \ref{fig9:opaicty_NFTS}, where $X=\{x_1,x_2,x_3,x_4\}=X_0$,
	$S=\{x_1\}$, $U=\{u_1,u_2\}$, and $Y=\{y_1,y_2,y_3\}$. It is not difficult to see that this NFTS
	is initial-state opaque, but not current-state opaque, or $K$-step for any positive integer $K$,
	or infinite-step opaque. Hence, InitSO does not imply KSO or CSO or InfSO.
	\begin{figure}
        \centering
\begin{tikzpicture}[>=stealth',shorten >=1pt,auto,node distance=2.5 cm, scale = 0.8, transform shape,
	->,>=stealth,inner sep=2pt,state/.style={shape=circle,draw,top color=red!10,bottom color=blue!30},
	point/.style={circle,inner sep=0pt,minimum size=2pt,fill=}, 
	skip loop/.style={to path={-- ++(0,#1) -| (\tikztotarget)}}]
	\node[initial,state,accepting,initial by arrow,initial where=left] (1) {$x_1/y_1$};
	\node[initial,state,initial by arrow,initial where=above] (2) [below of = 1] {$x_2/y_1$};
	\node[initial,state,initial by arrow,initial where=above] (3) [left of = 2] {$x_3/y_2$};
	\node[initial,state,initial by arrow,initial where=above] (4) [right of = 2] {$x_4/y_3$};

	\path [->]
	(3) edge node {$u_1$} (1)
	(3) edge node {$u_2$} (2)
	(3) edge [loop left] node {$u_1$} (3)
	(1) edge node {$u_1$} (4)
	(2) edge node {$u_1$} (4)
	(4) edge [loop right] node {$u_1$} (4)
	;

        \end{tikzpicture}
		\caption{State transition diagram of an NFTS in the proof of Theorem \ref{thm3:opacityNFTS}.}
		\label{fig9:opaicty_NFTS}
\end{figure}

	Fourth, consider the NFTS as in Fig. \ref{fig10:opaicty_NFTS}, where $X=\{x_1,x_2,x_3,x_4,x_5,x_6\}=X_0$,
	$S=\{x_1\}$, $U=\{u_1,u_2\}$, and $Y=\{y_1,y_2,y_3,y_4\}$.
	It is easy to see that the NFTS is $1$-step opaque and current-state opaque, but not initial-state opaque.
	Hence, neither KSO nor CSO implies InitSO.
	\begin{figure}
        \centering
\begin{tikzpicture}[>=stealth',shorten >=1pt,auto,node distance=2.5 cm, scale = 0.8, transform shape,
	->,>=stealth,inner sep=2pt,state/.style={shape=circle,draw,top color=red!10,bottom color=blue!30},
	point/.style={circle,inner sep=0pt,minimum size=2pt,fill=}, 
	skip loop/.style={to path={-- ++(0,#1) -| (\tikztotarget)}}]
	\node[initial,state,accepting,initial by arrow,initial where=left] (1) {$x_1/y_1$};
	\node[initial,state,initial by arrow,initial where=above] (2) [below of = 1] {$x_2/y_1$};
	\node[initial,state,initial by arrow,initial where=above] (3) [left of = 2] {$x_3/y_2$};
	\node[initial,state,initial by arrow,initial where=right] (4) [right of = 1] {$x_4/y_3$};
	\node[initial,state,initial by arrow,initial where=above] (5) [right of = 2] {$x_5/y_3$};
	\node[initial,state,initial by arrow,initial where=above] (6) [right of = 5] {$x_6/y_4$};

	\path [->]
	(3) edge node {$u_1$} (1)
	(3) edge node {$u_1$} (2)
	(3) edge [loop left] node {$u_1$} (3)
	(1) edge node {$u_1$} (4)
	(2) edge node {$u_1$} (5)
	(4) edge node {$u_2$} (6)
	(5) edge node {$u_1$} (6)
	(6) edge [loop right] node {$u_1$} (6)
	;

        \end{tikzpicture}
		\caption{State transition diagram of an NFTS in the proof of Theorem \ref{thm3:opacityNFTS}.}
		\label{fig10:opaicty_NFTS}
\end{figure}
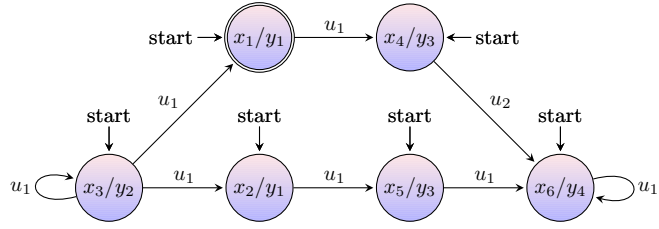

\end{proof}

\section{Verification of Opacity of NFTSs Using Two-Way Observers}\label{sec:checkingNFTS}

In Section \ref{sec:OpacityPreSimulation}, we propose several opacity-preserving (bi)simulation relations,
which could be used potentially to verify opacity for a class of infinite NTSs over their finite abstractions.
However, how can one verify opacity of finite abstractions? In this section, we use a two-way observer method
\cite{Yin2017TWObserverInfiniteStepOpacity} to verify various notions of opacity for NFTSs.
The two-way observer was proposed to verify infinite-step opacity and $K$-step opacity of DESs
in the framework of finite automata \cite{Yin2017TWObserverInfiniteStepOpacity}. To verify opacity of
NFTSs, we modify the method slightly. Next we introduce the technical details.

Note that the output function $h : X\to Y$ partitions $X$ into at most $|Y|$ observational equivalence classes.
For each $y\in Y$, we
denote by $X_y:= \{x\in X | h(x) = y\}$ the set of states whose output are $y$
and  denote by $X_{0,y}:= \{x\in X_0 | h(x) = y\}$ the set of initial states whose output are $y$.

Let $q\in 2^X$ be a set of states and $u\in U$ be an input.
We denote by
$\suc(q,u)$ the set of state that can be reached from $q$ under input $u$
and by
$\post(q,u)$ the set of state that can reach $q$ under input $u$, i.e.,
\begin{align}
\suc(q,u)&:=\{ x\in X| \exists x'\in q\text{ such that }(x',u,x)\in \to   \},\nonumber\\
\post(q,u)&:=\{ x\in X| \exists x'\in q\text{ such that }(x,u,x')\in \to   \}.
\end{align}

For an NFTS $(X,X_0,S,U,\to,Y,h)$, we define a new so-called verification NFTS (without secret states)
\begin{equation}\label{eqn2:opacity}
\Sigma_V=(X_V,X_{V,0},U_V,\to_V,Y_V,h_V),
\end{equation}
where
\begin{itemize}
  \item
	  $X_V\subseteq \{(q_1,q_2)\in 2^X\times 2^X|
	  \exists y_1,y_2\in Y\text{ such that }q_1\subseteq X_{y_1}\text{ and } q_2\subseteq X_{y_2}  \}$
  is the set of states;
  \item
  $X_{V,0}=
  \{ X_{0,y_1}\in 2^X| y_1\in Y\}
  \times
  \{X_{y_2} \in 2^X| y_2\in Y\}$
   is the set of initial states;
  \item
  $U_V=(U\times\{\epsilon\})\cup(\{\epsilon\}\times U)$ is the set of inputs;
  \item
  $\to_{V}\subseteq X_V\times U_V \times X_V$ is the transition relation
  defined as follows:
  For any $(q_1,q_2),(q_1',q_2')\in X_V$ and $u\in U$,
  \begin{itemize}
  \item
  $\left((q_1,q_2),(u,\epsilon),(q_1',q_2')\right)\in \to_V$
  if $q_2'=q_2$ and $\exists y\in Y$ such that $q_1'=\suc(q_1,u)\cap X_{y}\not=\emptyset$, and
  \item
  $\left((q_1,q_2),(\epsilon,u),(q_1',q_2')\right)\in \to_V$
  if $q_1'=q_1$ and $\exists y\in Y$ such that $q_2'=\post(q_2,u)\cap X_{y}\not=\emptyset$;
  \end{itemize}
  \item
  $Y_V=Y\times Y$ is the set of outputs;
  \item
  $h_V:X_V\to Y_V$ is defined
  for each $(q_1,q_2)\in X_V$ as
  $h_V((q_1,q_2))=(y_1,y_2)$,
  where $(y_1,y_2)$ is the unique pair such that $q_1\subseteq X_{y_1}$ and $q_1\subseteq X_{y_2}$.
  Particularly, we denote $h_{V,1}((q_1,q_2)):=y_1$ and $h_{V,2} (q_1,q_2)):=y_2$.
 \end{itemize}

For any given NFTS $\Sig$ as in Definition \ref{d4.1}, we construct the corresponding
NFTS $\Sig_V$ as in \eqref{eqn2:opacity}. For any given initial state $(q_0^1,q_0^2)$ of $\Sig_V$ in $X_{V,0}$,
an input sequence $\a=(u_0^1,u_0^2)\dots(u^1_{|\a|-1},u^2_{|\a|-1})$ in $(U_V)^*$, and states $(q_1^1,q_1^2),\dots,
(q_{|\a|}^1,q_{|\a|}^2)\in X_V$ such that
  \[
  (q_0^1,q_0^2)\xrightarrow{(u_0^1,u_0^2)}_V \cdots \xrightarrow{(u_{|\a|-1}^1,u_{|\a|-1}^2)}_V(q_{|\a|}^1,q_{|\a|}^2),
  \]
we have that the left component $q_0^1\xrightarrow{u_0^1}_V \cdots \xrightarrow{u_{|\a|-1}^1}_Vq_{|\a|}^1$
aggregates all runs of $\Sig$ starting from some initial state of $q_0^1$ over $u_0^1\dots u^1_{|\a|-1}$
and producing the output sequence $h_{V,1}( (q_0^1,q_0^2))\dots h_{V,1}((q_{|\a|}^1,q_{|\a|}^2))$
(note that repetition of states of the form $x\xrightarrow[]{\e}x$ may exist),
and the right component $q_0^2\xrightarrow{u_0^2}_V \cdots \xrightarrow{u_{|\a|-1}^2}_Vq_{|\a|}^2$
aggregates the mirror images of all runs of $\Sig$ ending at some state of $q_{0}^2$ over $u_{|\a|-1}^2\dots u^2_{0}$
and producing the output sequence $h_{V,2}( (q_{|\a|}^1,q_{|\a|}^2))\dots h_{V,2}((q_0^1,q_0^2))$
(note that repetition of states of the form $x\xrightarrow[]{\e}x$ may also exist).
Based on this direct observation and preliminary definitions, the following proposition holds.

\begin{proposition}\label{prop:main}
	For NFTS \eqref{eqn2:opacity}, for any input sequence $\a=(u_0^1,u_0^2)\dots (u_{|\a|-1}^1,u_{|\a|-1}^2)\in U_V^*$,
and any transitions
  \[
  (q_0^1,q_0^2)\xrightarrow{(u_0^1,u_0^2)}_V \cdots \xrightarrow{(u_{|\a|-1}^1,u_{|\a|-1}^2)}_V(q_{|\a|}^1,q_{|\a|}^2),
  \]
where $(q_0^1,q_0^2)\in X_{V,0}$, we have:
\begin{enumerate}
  \item
  $q_{|\a|}^1=\{x_{|\a|} \in X| \exists x_0\in q_0^1\text{ such that }
  x_0\xrightarrow{u_0^1} \cdots \xrightarrow{u_{|\a|-1}^1} x_{|\a|}\text{ and }
  \forall i\in [0,|\a|],h(x_i)=h_{V,1}((q_{i}^1,q_{i}^2)) \}$;
  \item
  $q_{|\a|}^2=\{x_{0} \in X| \exists x_{|\a|}\in q_{0}^2 \text{ such that }
  x_0\xrightarrow{u_{|\a|-1}^2} \cdots \xrightarrow{u_{0}^2} x_{|\a|}\text{ and }
  \forall i\in [0,|\a|],h(x_{|\a|-i})=h_{V,2}((q_{i}^1,q_{i}^2)) \}$.
\end{enumerate}
\end{proposition}
\begin{proof}
We prove, by induction on the length of $\a$, that (1) is true. The proof of (2) is similar to (1).
When $|\a|=0$, by definition,  we know that  $q_0^1=X_{0,y}$ for some $y\in Y$ and we have $h_{V,1}((q_0^1,q_0^2))=y$.
Therefore, we know that
\[
q_{0}^1=\{x_0 \in X| \exists x_0\in X_0\text{ such that }h(x_0)=h_{V,1}((q_{0}^1,q_{0}^2))\},
\]
i.e.,  the induction basis holds.

To proceed the induction, now let us assume that (1) holds for input sequence $\a=(u_0^1,u_0^2)\dots (u_{n-1}^1,u_{n-1}^2)\in U_V^*$
and we need to show that (1) still holds for input sequence $\a=(u_0^1,u_0^2)\dots (u_{n-1}^1,u_{n-1}^2)(u_{n}^1,u_{n}^2)\in U_V^*$.
For $(q_{n}^1,q_{n}^2)\xrightarrow{(u_{n}^1,u_{n}^2)}(q_{n+1}^1,q_{n+1}^2)$, we consider the following two cases:
$u_{n}^1=\epsilon$ or $u_{n}^1\not=\epsilon$.
If $u_{n}^1=\epsilon$, then we know that $u_{0}^1\dots u_{n-1}^1=u_{0}^1\dots u_{n-1}^1u_{n}^1$
and $q_{n}^1=q_{n+1}^1$. Therefore, the induction step holds immediately.
Hereafter, we consider the case that $u_{n}^1\not=\epsilon$.
By the definition of $\to_V$, we know that
$q_{n+1}=\suc(q_n,u)\cap X_{y}$ for some $y_{n+1}\in Y$, i.e.,
\begin{align*}
q_{n+1}^1=\{& x_{n+1}\in X| \exists x_n\in q_{n}\text{ s.t. }  x_n\xrightarrow{u_{n}} x_{n+1}\wedge h(x_{n+1})=y_{n+1}  \}
\end{align*}
and for any $x_{n+1}\in q_{n+1}$, we have
$h(x_{n+1})=h_{V,1}((q_{n+1}^1,q_{n+1}^2))=y_{n+1}$.
This together with the induction hypothesis implies that
\begin{align}
q_{n+1}^1=\left\{
	x_{n+1} \in X
\left|
\begin{array}{c c c}
&\exists x_0\in X_0\text{ s.t. }   \\
&x_0\xrightarrow{u_0^1} \cdots \xrightarrow{u_{n-1}^1} x_{n}\xrightarrow{u_{n}^1} x_{n+1} \wedge\\
&\forall i\in [0,n+1]:  h(x_i)\!=\!h_{V,1}((q_{i}^1,q_{i}^2))
\end{array}\right.
\!\!\!\right\},
\end{align}
which completes the proof.
\end{proof}

By Proposition \ref{prop:main}, we obtain the following four theorems used for verifying
the four types of opacity for NFTSs.

\begin{theorem}\label{thm1_opacity_TwoWayObserver}
NFTS $\Sig=(X,X_0,S,U,\to,Y,h)$ is  current-state opaque if and only if
\begin{equation}
\forall (q^1,q^2)\in X_V, q^1\not\subseteq S.
\end{equation}
\end{theorem}
\begin{proof}($\Rightarrow$)
By contraposition:
suppose that there exists $(q^1,q^2)\in X_V$ such that $q^1\subseteq S$.
Let
  \[
  (q_0^1,q_0^2)\xrightarrow{(u_0^1,u_0^2)}_V \cdots \xrightarrow{(u_{n-1}^1,u_{n-1}^2)}_V(q_{n}^1,q_{n}^2)
  \]
be a sequence reaching $(q^1_n,q^2_n)=(q^1,q^2)$.
Then we consider the following sequence
$x_0\xrightarrow{u_0^1}\cdots \xrightarrow{u_{n-1}^1}x_n$ in $\Sigma$ such that $x_n\in S$ and  $h(x_i)=h_{V,1}((q_{i}^1,q_{i}^2)),\forall i=0,\dots,n$.
Since $q^1_n\subseteq S$,   by Proposition~\ref{prop:main},
for any transitions $x_0'\xrightarrow{u_0^1}\cdots \xrightarrow{u_{n-1}^1}x_n'$ in $\Sigma$
such that  $h(x_i')=h(x_i),\forall i=0,\dots,n$, we have $x_{n}'\in S$.
This implies that $\Sigma$ is not current-state opaque.

($\Leftarrow$)
By contraposition:
suppose that $\Sig$ is not current-state opaque.
Then we know that
there exists $x_0\xrightarrow{u_0}\cdots \xrightarrow{u_{n-1}}x_n$ in $\Sigma$,
such that
 (i) $x_{n}\in S$;  and
(ii) for any $x_0'\xrightarrow{u_0}\cdots \xrightarrow{u_{n-1}}x_n'$ in $\Sigma$
such that  $h(x_i)=h(x_i'),\forall i=0,\dots,n$, we have $x_{n}'\in S$.
Then let us consider the following sequence in $\Sig_V$
  \[
  (q_0^1,q_0^2)\xrightarrow{(u_0,\epsilon)}_V \cdots \xrightarrow{(u_{n-1},\epsilon)}_V(q_{n}^1,q_{n}^2)
  \]
where $h_{V,1}(q_i^1,q_i^2)=h(x_i)=h(x_i'),\forall i=0,\dots,n$.
By Proposition~\ref{prop:main}, we know that
 $q_{n}^1=\{x_{n} \in X| \exists x_0\in q_0^1,x_0\xrightarrow{u_0} \cdots \xrightarrow{u_{n-1}} x_{n}  \text{ and }
  \forall i\in [0,n],  h(x_i)=h_{V,1}((q_{i}^1,q_{i}^2))  \}$.
This together with (ii) above imply that $q_n^1\subseteq S$.
\end{proof}

\begin{theorem}\label{thm2_opacity_TwoWayObserver}
NFTS $\Sig=(X,X_0,S,U,\to,Y,h)$ is  initial-state opaque if and only if
\begin{equation}
	\forall (q^1,q^2)\in X_V,q^2\cap X_0\ne\emptyset\Rightarrow q^2\cap X_0\not\subseteq S.
\end{equation}
\end{theorem}
\begin{proof}
($\Rightarrow$)
By contraposition:
suppose that there exists $(q^1,q^2)\in X_V$ such that $q^2\cap X_0\ne\emptyset$ and $ q^2\cap X_0 \subseteq S$.
Let
  \[
  (q_0^1,q_0^2)\xrightarrow{(u_0^1,u_0^2)}_V \cdots \xrightarrow{(u_{n-1}^1,u_{n-1}^2)}_V(q_{n}^1,q_{n}^2)
  \]
be a sequence reaching $(q^1_n,q^2_n)=(q^1,q^2)$.
Then we consider  sequence
$x_0\xrightarrow{u_{n-1}^2}\cdots \xrightarrow{u_{0}^2}x_n$ in $\Sigma$ such that $x_0\in S$ and  $h(x_{n-i})=h_{V,2}((q_{i}^1,q_{i}^2)),\forall i=0,\dots,n$.
Such a sequence is well defined since $q_{i+1}^2\subseteq \post(q_i^2,u_i^2)$.
Then by Proposition~\ref{prop:main}, we know that
for any $x_0'\xrightarrow{u_{n-1}^2} \cdots \xrightarrow{u_{0}^2} x_{n}' $ in $\Sigma$
such that  $h(x_i')=h(x_i)=h_{V,2}((q_{n-i}^1,q_{n-i}^2)),\forall i=0,\dots,n$, we have $x_{0}'\in S$.
This implies that $\Sigma$ is not initial-state opaque.

($\Leftarrow$)
By contraposition:
suppose that $\Sig$ is not initial-state opaque.
Then we know that
there exists $x_0\xrightarrow{u_0}\cdots \xrightarrow{u_{n-1}}x_n$ in $\Sigma$,
such that
 (i) $x_{0}\in S$;  and
(ii) for any $x_0'\xrightarrow{u_0}\cdots \xrightarrow{u_{n-1}}x_n'$ in $\Sigma$
such that  $h(x_i)=h(x_i'),\forall i=0,\dots,n$, we have $x_{0}'\in S$.
Then let us consider the following sequence in $\Sig_V$
  \[
  (q_0^1,q_0^2)\xrightarrow{(\epsilon,u_{n-1})}_V \cdots \xrightarrow{(\epsilon,u_{0})}_V(q_{n}^1,q_{n}^2)
  \]
where $h_{V,2}(q_{n-i}^1,q_{n-i}^2)=h(x_i)=h(x_i'),\forall i=0,\dots,n$.
By Proposition~\ref{prop:main}, we know that
$q_{n}^2=\{x_{0} \in X: \exists x_{n}\in q_{0}^2 \text{ such that }
 x_0\xrightarrow{u_{n-1}^2} \cdots \xrightarrow{u_{0}^2} x_{n}\text{ and }
 \forall i\in [0,n],  h(x_{n-i})=h_{V,2}((q_{i}^1,q_{i}^2))\}$.
This together with (ii) above imply that $q_n^2\cap X_0\ne\emptyset$ and $q_n^2\cap X_0 \subseteq S$.
\end{proof}

\begin{remark}
Let us discuss the complexity for the verifications of current-state opacity
and initial-state opacity using the above theorems.
In the worst case, $\Sig_V$ contains at most $4^{|X|}$ states and $2|Y||U|4^{|X|}$ transitions.
Also, we note that  $\Sig_V$ is a pure shuffle in the sense that its first and its second components are independent.
Therefore, to verify current-state opacity (respectively, initial-state opacity),
we just need to construct the first component (respectively, the second component) of $\Sig_V$.
Hence, the time complexity for the verifications of current-state opacity and initial-state opacity are
both $O(|Y||U|2^{X})$.
\end{remark}

\begin{theorem}\label{thm3_opacity_TwoWayObserver}
NFTS $\Sig=(X,X_0,S,U,\to,Y,h)$ is infinite-step opaque if and only if
\begin{equation}
\forall (q^1,q^2)\in X_V, q^1\cap q^2\neq\emptyset \Rightarrow q^1\cap q^2\not\subseteq S.
\end{equation}
\end{theorem}
\begin{proof}($\Rightarrow$)
By contraposition:
suppose that there exists $(q^1,q^2)\in X_V$ such that
$\emptyset\not= q^1\cap  q^2 \subseteq S$.
Let
  \[
  (q_0^1,q_0^2)\xrightarrow{(u_0^1,u_0^2)}_V \cdots \xrightarrow{(u_{n-1}^1,u_{n-1}^2)}_V(q_{n}^1,q_{n}^2)
  \]
be a sequence reaching $(q^1_n,q^2_n)=(q^1,q^2)$.
Then we consider  sequence
\[
x_0\xrightarrow{u_{0}^1}\cdots \xrightarrow{u_{n-1}^1}x_n
\xrightarrow{u_{n-1}^2}\cdots \xrightarrow{u_{0}^2}x_{2n}
\]
in $\Sigma$ such that $x_n\in q^1\cap  q^2 \subseteq S$,
$h(x_i)=h_{V,1}((q_{i}^1,q_{i}^2)),\forall i=0,\dots,n$ and
$h(x_{2n-i})=h_{V,2}((q_{i}^1,q_{i}^2)),\forall i=0,\dots,n$.
This sequence is well defined since $\emptyset\not= q^1\cap  q^2 \subseteq S$.
Then by Proposition~\ref{prop:main}, we know that
for any $x_0'\xrightarrow{u_{0}^1}\cdots \xrightarrow{u_{n-1}^1}x_n'
\xrightarrow{u_{n-1}^2}\cdots \xrightarrow{u_{0}^2}x_{2n}'$ in $\Sigma$
such that  $h(x_i')=h(x_i),\forall i=0,\dots,2n$, we have $x_{n}'\in S$.
This implies that $\Sigma$ is not infinite-step opaque.

($\Leftarrow$)
By contraposition:
suppose that $\Sig$ is not infinite-step opaque.
Then we know that
there exists $x_0\xrightarrow{u_0}\cdots \xrightarrow{u_{n-1}}x_n\xrightarrow{u_n}\cdots \xrightarrow{u_{n+m-1}}x_{n+m}$ in $\Sigma$,
such that
 (i) $x_{n}\in S$;  and
(ii) for any $x_0'\xrightarrow{u_0}\cdots \xrightarrow{u_{n-1}}x_n'\xrightarrow{u_n}\cdots \xrightarrow{u_{n+m-1}}x_{n+m}'$ in $\Sigma$
such that  $h(x_i)=h(x_i'),\forall i=0,\dots,n$, we have $x_{n}'\in S$.
Then let us consider the following sequence in $\Sig_V$
\begin{align}
  (q_0^1,q_0^2)
  \xrightarrow{(u_{0},\epsilon)}_V \cdots \xrightarrow{(u_{n-1},\epsilon)}_V(q_{n}^1,q_{n}^2)
  \xrightarrow{(\epsilon,u_{n+m-1})}_V \cdots \xrightarrow{(\epsilon,u_{n})}_V(q_{n+m}^1,q_{n+m}^2),
\end{align}
where $h_{V,1}(q_i^1,q_i^2)=h(x_i),\forall i=0,\dots,n$
and $h_{V,2}(q_{n+i}^1,q_{n+i}^2)=h(x_{n+m-i}),\forall i=0,\dots,m$.
By Proposition~\ref{prop:main}, we know that
$q_{n+m}^1=
\{x_{n} \in X| \exists x_0\in q_0^1,x_0\xrightarrow{u_0} \cdots \xrightarrow{u_{n-1}} x_{n}  \text{ and }
\forall i\in [0,n],  h(x_i)=h_{V,1}((q_{i}^1,q_{i}^2))  \}$,
and
$q_{n+m}^2=\{x_{n} \in X| \exists x_{n+m}\in q_{0}^2 \text{ and }
x_n\xrightarrow{u_{n}^2} \cdots \xrightarrow{u_{n+m-1}^2} x_{n+m}\text{ and }
  \forall i\in [0,m],  h(x_{n+i})=h_{V,2}((q_{n+m-i}^1,q_{i}^2)) \}$.
This together with (ii) above imply that $\emptyset\not= q_n^1\cap q_n^2\subseteq S$.
\end{proof}

\begin{theorem}\label{thm4_opacity_TwoWayObserver}
NFTS $\Sig=(X,X_0,S,U,\to,Y,h)$ is  \emph{not}  $K$-step opaque if and only if there exists a sequence
$x_{V,0}\xrightarrow{(u_0^1,u_0^2)\dots (u_{n-1}^1,u_{n-1}^2)}_V(q^1,q^2)$, where $x_{V,0}\in X_{V,0}$ such that
\begin{equation}
\emptyset \neq q^1\cap q^2 \subseteq S
\text{ and }
|u_0^2\dots u_{n-1}^2|\leq K.
\end{equation}
\end{theorem}
\begin{proof}
The proof is similar to the proof of Theorem~\ref{thm3_opacity_TwoWayObserver}.
Specifically, in the ``$\Rightarrow$" direction, since $|u_0^2\dots u_{n-1}^2|\leq K$ and
infinite-step opacity is stronger than $K$-step opacity, one concludes ``$\Rightarrow$" direction from Theorem~\ref{thm3_opacity_TwoWayObserver}.
Similarly, in the ``$\Leftarrow$" direction,
the violation of $K$-step opacity allows us to choose $m$ such that $m\leq K$.
\end{proof}
\begin{remark}
To verify infinite-step opacity, we need to construct automaton $\Sig_V$ completely for both components.
Hence, the complexity is $O(|Y||U|4^{|X|})$.
To verify $K$-step opacity, we need to construct
parts of $\Sig_V$ that can be reached from initial states within $K$-steps in the second component.
Therefore, the complexity for verifying $K$-step opacity is
$O(\min\{ 2^{|X|},(|U||Y|)^K\}|U||Y|2^{|X|})$.

\end{remark}

\section{Example}
\label{sec:application}

In this section, we show an example to illustrate how to use the main results of this paper to verify
the opacity of an infinite transition system.

Consider the following discrete-time control system
\begin{equation}\label{eqn30:opacity}
	\begin{split}
		x_1(t+1) =& f_1(x_1(t),x_2(t)), \\
		x_2(t+1) =& f_2(x_1(t),x_2(t)), \\
	y(t) =& h(x_1(t),x_2(t)),
	\end{split}
\end{equation}
where $t\in\N$; $x_1(t),x_2(t),y(t)\in\R$;
\begin{align*}
	\begin{bmatrix}f_1(x_1,x_2)\\f_2(x_1,x_2)\end{bmatrix}&=\left\{
\begin{array}[]{ll}
	\begin{bmatrix}
	0 & -2\\ 2 & 0
\end{bmatrix}\begin{bmatrix}
		x_1\\ x_2
	\end{bmatrix}&\text{ if }x_1> 0,x_2\ge 0,\\
	\begin{bmatrix}
	0 & -\frac{1}{2}\\ \frac{1}{2} & 0
\end{bmatrix}\begin{bmatrix}
		x_1\\ x_2
	\end{bmatrix}&\text{ if }x_1\le 0,x_2> 0,\\
	\begin{bmatrix}
	0 & -1\\ 1 & 0
\end{bmatrix}\begin{bmatrix}
		x_1\\ x_2
	\end{bmatrix}&\text{ otherwise,}
\end{array}\right.
\end{align*}
\begin{align}\label{eqn6:opacity}
	h(x_1,x_2)&=\left\{
	\begin{array}[]{ll}
		1 &\text{ if }(x_1,x_2)\in A_1\cup A_2,\\
		2 &\text{ if }(x_1,x_2)\in B_1\cup B_2,\\
		3 &\text{ if }(x_1,x_2)\in C_1\cup C_2,\\
		4 &\text{ if }(x_1,x_2)\in D_1\cup D_2,\\
		5 &\text{ if }(x_1,x_2)\in E,
	\end{array}
	\right.
\end{align}
\begin{align*}
	A_1 &= \left\{(x_1,x_2)\in\R^2\left|0<x_1\le \frac{1}{2},0\le x_2\le \frac{1}{2}\right.\right\},\\
	A_2 &= \left\{(x_1,x_2)\in\R^2\left|0<x_1\le 1,0\le x_2\le 1\right.\right\}\setminus A_1,\\
	B_1&= \left\{(x_1,x_2)\in\R^2\left| -1\le x_1\le 0,0<x_2\le 1 \right.\right\},\\
	B_2&= \left\{(x_1,x_2)\in\R^2\left| -2\le x_1\le 0,0<x_2\le 2 \right.\right\}\setminus B_1,\\
	C_1 &= \left\{(x_1,x_2)\in\R^2\left|-\frac{1}{2}\le x_1<0 ,-\frac{1}{2}\le x_2\le 0\right.\right\},\\
	C_2 &= \left\{(x_1,x_2)\in\R^2\left|-1\le x_1<0,-1\le x_2\le 0\right.\right\}\setminus C_1,\\
	D_1 &= \left\{(x_1,x_2)\in\R^2\left|0\le x_1\le \frac{1}{2} ,-\frac{1}{2}\le x_2< 0\right.\right\},\\
	D_2 &= \left\{(x_1,x_2)\in\R^2\left|0\le x_1\le 1,-1\le x_2<0\right.\right\}\setminus D_1,\\
	E &= \R^2\setminus(A_1\cup A_2\cup B_1\cup B_2\cup C_1\cup C_2\cup D_1\cup D_2).
\end{align*}

If we let each state of \eqref{eqn30:opacity} be initial and choose the secret state set $A_1=:S$, then
\eqref{eqn30:opacity} can be written as the following NTS
\begin{equation}
	\Sig_{A_1}=(\R^2,\R^2,A_1,\{u\},\to,Y,h),
\end{equation}
where for all $x,x'\in \R^2$,
$(x,u,x')\in\to$ if and only if $x'=f(x)$, $f=(f_1,f_2)$; $Y=\{1,2,3,4,5\}$; and $h$ is as in \eqref{eqn6:opacity}.

Next, we use the main results obtained in this paper to verify the opacity of this system.
We define an equivalence relation $\sim$ on $\R^2$ for all $x,x'\in\R^2$ as $(x,x')\in\sim$ if and only
if $x$ and $x'$ both belong to the same one of $A_1,A_2,B_1,B_2,C_1,C_2,D_1,D_2$, and $E$.
Then $\sim$ is an InfSOP bisimulation relation between $\Sig_{A_1}$ and itself
according to Definition \ref{def_OpacityPreservingBisimulation}.
By Theorem \ref{thm2:opacityNFTS}, the corresponding quotient relation $\sim_\mathsf{Q}=\{(x,[x])|x\in\R^2\}$
is an InfSOP bisimulation relation between $\Sig_{A_1}$ and its quotient system
\begin{equation}\label{eqn5:opacity}
	\Sig_{A_1\sim}=(X,X,\{A_1\},\{u\},\to_{\sim},Y,h_{\sim}),
\end{equation}
where $X=\{A_1,A_2,B_1,B_2,C_1,C_2,D_1,D_2,E\}$,
$\to_{\sim}=\{(A_i,u,B_i),(B_i,u,C_i),(C_i,u,D_i),(D_i,u,A_i),(E,u,E)|i=1,2\}$,
for each $\bar x\in X$, $h_{\sim}(\bar x)=h(x)$, where $x\in \bar x$.

It is not difficult to obtain that system $\Sig_{A_1\sim}$ is infinite-step opaque,
so $\Sig_{A_1}$ is also infinite-step opaque by Theorem \ref{thm1:opacityNFTS}.
Then by Theorem \ref{thm3:opacityNFTS}, it is also initial-state opaque, current-state opaque, and $K$-step opaque
for any positive integer $K$.

\section{Conclusion}\label{sec:conc}

In this paper, we proposed several notions of opacity-preserving (bi)simulation relations
from an NTS to another NTS, and used quotient system construction to potentially compute such relations.
Hence, although the verification of opacity of NTSs is generally undecidable,
if we find such a relation from an NTS to an NFTS, we can
verify the opacity of the NTS over the NFTS which is decidable. We also propose a two-way observer method to verify
the opacity of NFTSs.
In addition, we verify the opacity of an infinite transition system to illustrate the main results in Section
\ref{sec:application}.

Although the construction of proposed relations here based on quotient systems can be used to deal with some classes of
NTSs, generally
it is not easy to check the existence of appropriate quotient relations implementing them. So in order to make these
opacity-preserving (bi)simulation relations applicable to more classes of NTSs, in the future we will investigate different algorithms on the construction of NFTSs for NTSs.


\end{document}